\documentclass[lettersize,journal]{IEEEtran}

\usepackage{amsmath,amsfonts,amsthm,amssymb,optidef,stmaryrd}
\usepackage{mdframed}
\usepackage{algorithmic}
\usepackage{algorithm}
\usepackage{array}
\usepackage[caption=false,font=normalsize,labelfont=sf,textfont=sf]{subfig}
\usepackage{textcomp}
\usepackage{stfloats}
\usepackage{url}
\usepackage{verbatim}
\usepackage{graphicx}
\usepackage{cite}
\usepackage{color}
\usepackage{multirow,booktabs,makecell}
\usepackage[hidelinks]{hyperref}

\usepackage{flushend}

\mdfsetup{%
innertopmargin=0pt,
innerleftmargin=3pt,
innerbottommargin=3pt,
innerrightmargin=3pt
}

\definecolor{mygreen}{RGB}{11, 218, 81}
\definecolor{myblue}{RGB}{31, 81, 255}
\definecolor{mypurple}{RGB}{191, 64, 191}
\definecolor{myorange}{RGB}{255, 172, 28}

\newcommand{\abs}[1]{\ensuremath{\left| #1 \right|}}

\newcommand{\norm}[2]{\ensuremath{\left\lVert #1 \right\rVert}_#2}

\newcommand*{\transp}{{\mkern-1.5mu\mathsf{T}}}

\DeclareMathOperator*{\argmin}{arg\,min}

\newtheorem{remark}{Remark}[section]
\newtheorem{proposition}{Proposition}
\newtheorem{lemma}{Lemma}
\newtheorem{theorem}{Theorem}
\newtheorem{corollary}{Corollary}

\begin{document}

\title{Model-based learning for multi-antenna multi-frequency location-to-channel mapping}

\author{
	Baptiste Chatelier, Vincent Corlay, Matthieu Crussière, Luc Le Magoarou
    \thanks{Baptiste Chatelier is with Mitsubishi Electric R\&D Centre Europe, Univ Rennes, INSA Rennes, CNRS, IETR-UMR 6164 and b\raisebox{0.2mm}{\scalebox{0.7}{\textbf{$<>$}}}com, Rennes, France (email: baptiste.chatelier@insa-rennes.fr).}
    \thanks{Vincent Corlay is with Mitsubishi Electric R\&D Centre Europe and b\raisebox{0.2mm}{\scalebox{0.7}{\textbf{$<>$}}}com, Rennes, France (email: v.corlay@fr.merce.mee.com).}
    \thanks{Matthieu Crussière and Luc Le Magoarou are with Univ Rennes, INSA Rennes, CNRS, IETR-UMR 6164 and b\raisebox{0.2mm}{\scalebox{0.7}{\textbf{$<>$}}}com, Rennes, France (email: \{matthieu.crussiere ; luc.le-magoarou\}@insa-rennes.fr).}
	}



\maketitle

\begin{abstract}
    Years of study of the propagation channel showed a close relation between a location and the associated communication channel response. The use of a neural network to learn the location-to-channel mapping can therefore be envisioned. The Implicit Neural Representation (INR) literature showed that classical neural architecture are biased towards learning low-frequency content, making the location-to-channel mapping learning a non-trivial problem. Indeed, it is well known that this mapping is a function rapidly varying with the location, on the order of the wavelength. This paper leverages the model-based machine learning paradigm to derive a problem-specific neural architecture from a propagation channel model. The resulting architecture efficiently overcomes the spectral-bias issue. It only learns low-frequency sparse correction terms activating a dictionary of high-frequency components. The proposed architecture is evaluated against classical INR architectures on realistic synthetic data, showing much better accuracy. Its mapping learning performance is explained based on the approximated channel model, highlighting the explainability of the model-based machine learning paradigm.
\end{abstract}


\section{Introduction}\label{sec:introduction}
    \IEEEPARstart{F}{or} the past decades, signal processing methods have been used to improve communication systems. Such methods are model-based: they can present a high bias but benefit from a reasonable complexity. With the emergence of easily accessible computational power, artificial intelligence (AI)/machine learning (ML) has emerged as a promising alternative to signal processing methods in many communication problems. By essence, AI/ML methods are data-based: they consequently present a low bias due to their intrinsic adaptability capabilities. However, the training prerequisites of such methods entail substantial computational and sample complexities. Recently, researchers have focused on bridging the gap between those two paradigms using a hybrid approach: model-based machine learning~\cite{Shlezinger2023}. This approach proposes to use models from signal processing, to structure, initialize, and train learning methods from AI/ML. The underlying goal is to reduce the bias of signal processing methods by making models more flexible, while guiding AI/ML methods to reduce their complexity.
    
    In the past few years, model-based machine learning have showed great results in channel estimation problems~\cite{Hengtao18,Xiuhong21,yassine2022,Chatelier2022}, angle of arrival estimation~\cite{Shmuel2023,chatelier24_diffMUSIC}, channel charting~\cite{Yassine2022a,yassine2023modelbased,yassine2023charting}, but also in integrated sensing and communication scenarios~\cite{mateosramos2023semisupervised}. It is proposed to further explore this paradigm by studying the location-to-channel mapping learning: as the propagation channel coefficients are closely related to the user's location, one can use a neural network to learn this specific mapping. Upon training completion, one only has to input a location to the neural network to acquire the channel coefficients at the given location. In order to do so, one can use a physical propagation model to derive a model-based (MB) neural architecture specifically adapted for the location-to-channel mapping learning.

    This approach would be beneficial in many applications: channel estimation, secure communication mechanisms, resource allocation, interference management, and also radio-environment compression. Indeed, if one achieves near perfect learning of the location-to-channel mapping, it could be more efficient to only store the weights of the trained neural network rather than directly storing the channel coefficients.

    Learning continuous mappings with neural networks is known as the Implicit Neural Representation (INR) problem. After its great success in the resolution of image processing problems such as novel view synthesis~\cite{mildenhall2020nerf}, researchers have started to establish theoretical results on the INRs mapping learning capabilities~\cite{Ramasinghe2021,Zheng2021,Yuce2022,Ramasinghe2022a,Saratchandran2024,Saratchandran2024a}. While classical architectures, such as multi-layer perceptrons (MLPs), are universal function approximators~\cite{Cybenko1989ApproximationBS,Hornik89}, it has been shown that they exhibit a bias towards learning low-frequency functions, a phenomenon known as the spectral bias~\cite{Rahaman2019,ijcai2021p304,JohnXu2020}. This makes classical architectures unsuitable for the learning of rapidly varying functions. To address this limitation, several specialized architectures have been developed: random Fourier features (RFFs)~\cite{rahimi2007,tancik2020fourfeat}, sinusoidal representation networks (SIRENs)~\cite{sitzmann2019siren} or Gaussian activated radiance fields (GARFs)~\cite{Chng2022}. These architecture aim to incorporate high-frequency content into the neural model. The principles underlying these spectral-bias-resistant designs can be categorized into two approaches: embedding specialization within traditional \texttt{ReLU}-MLPs or embedding replacement.
    \begin{itemize}
        \item \noindent\textbf{Embedding specialization.} This approach consists on the projection of the input on a higher dimensional space containing high-frequency content, and is commonly referred to as the \textit{positional embedding} method. This approach, firstly introduced in~\cite{mildenhall2020nerf}, aligns with earlier findings in~\cite{Rahaman2019}, where the authors showed that introducing a particular high-frequency embedding allowed to ease the learning of high-frequency content in the target mapping. Further work focused in the  design of a well adapted embedding layer for the target mapping learning, while the rest of the neural architecture is a traditional \texttt{ReLU}-MLP. This approach has been used in~\cite{rahimi2007,mildenhall2020nerf,tancik2020fourfeat,Pumarola2020,Park2020,Du2021}.
        \item \noindent\textbf{Embedding replacement.} Conversely, this approach suggests omitting the initial high-frequency embedding and instead specializing the activation functions of the MLP. This method recently garnered significant attention: in~\cite{sitzmann2019siren}, the authors proposed to replace \texttt{ReLU} non-linearities by sine functions, achieving high performance in high-frequency mapping learning. Nevertheless, it has been shown that using sine non-linearities can lead to high sensitivity to the network parameters initialization. To address this, the authors in~\cite{Chng2022} proposed replacing sine non-linearities with Gaussian ones, achieving strong performance in image reconstruction problems without requiring intricate network parameters initialization schemes. Additionally, in~\cite{Ramasinghe2021}, the authors showed that sine non-linearities were part of a broad class of non-linearities that permitted MLPs to learn high-frequency content. They also showed that using non-periodic non-linearities allowed to obtain good performance in novel view synthesis problems. Finally, it has recently been shown in~\cite{Saragadam2023} that complex Gabor wavelets could be used as activation functions, yielding good performance in a wide range of image processing problems.
    \end{itemize}
    The location-to-channel mapping presents a high-frequency spatial dependence, on the order of the operating wavelength, making its learning a remarkably complex problem. One may then pose the subsequent inquiries: \textit{Are classical INR architectures able to learn the location-to-channel mapping? Is a model-based approach able to learn this mapping? Does the model-based approach outperform INR architectures in terms of learning performance and complexity?}

    \noindent\textbf{Contributions.} It is proposed to leverage the capabilities of model-based machine learning for the location-to-channel mapping learning. To maintain generality, this mapping learning problem is examined over multiple antennas and multiple frequencies. The following contributions are presented in this paper:
        \begin{itemize}
        \item A theoretical study of a physics-based channel model, using Taylor expansions, allows to clearly separate low and high-frequency spatial contents (Lemma~\ref{lemma:taylor_approx}, Proposition~\ref{proposition:channel_approx}, Proposition~\ref{proposition:matrix_approximation}).
        \item A sparse signal processing approach, using sparse representations of the channel, is presented to ensure the validity of the obtained channel approximation across the entire $\mathbb{R}^3$ space (Theorem~\ref{theorem:global_validity_approx}).
        \item A model-based neural architecture is derived from the approximated channel model. It is shown that the obtained structure shares common features with classical INR architectures, such as a spectral separation stage with a high-frequency embedding, allowing to bypass the spectral-bias issue.
        \item Experiments are conducted on realistic channels, showing the great potential of the proposed approach, both in terms of learning performance and computational efficiency. Moreover, some experiments clearly show that the proposed approach successfully learns the underlying physical models, enabling efficient generalization.
    \end{itemize}

    \noindent\textbf{Related work.} Learning continuous mappings through neural networks has been extensively studied by the image processing community. For instance, it has been used for image reconstruction problems~\cite{sitzmann2019siren,Bemana2020,Chen2021}, for 3D scene reconstruction from 2D images~\cite{mildenhall2020nerf,Yariv2020,Sitzmann2019} but also for dynamic 3D scene reconstruction from 2D images~\cite{Pumarola2020,Park2020,Tretschk2021,Du2021,Xian2021}. In the signal processing community, machine learning capabilities have been recently leveraged for various problems: e.g. ML-based channel estimation~\cite{Xuanxuan18,Mehran19,Eren20,Yu2020}. There also exist works about mapping learning in communication problems: specifically, the location/pseudo-location-to-beamformer mapping learning has been studied in~\cite{lemagoarou2022} and~\cite{lemagoarou2022_asilomar}, as well as the pseudo-location-to-best-codebook-index mapping learning in~\cite{yassine2023modelbased}. Additionally, our previous work in~\cite{chatelier2023modelbased} studied the location-to-channel mapping learning in a simplified scenario: only a single antenna and frequency were considered. This paper builds on that foundation, offering a significant extension of the previously proposed method by incorporating support for both multi-antenna and multi-frequency scenarios, thereby addressing broader and more complex use cases. In~\cite{Zhao2023}, the authors propose to adapt the neural radiance fields concept of~\cite{mildenhall2020nerf} from optical to radio-frequency signals. However, the proposed model requires complex learning strategies and long training times. Additionally, the proposed model is used for downlink channel estimation in a frequency division duplex setting, but requires the uplink channel knowledge. In this paper, the proposed approach only requires the receiver location knowledge and possesses relatively low sample complexity. Finally, in~\cite{Hoydis2023}, the authors proposed to calibrate electromagnetic properties of a scene, e.g. material permittivity and conductivity, scattering and reflection coefficients, using a differentiable ray-tracing approach. Upon completion of the calibration process, accessing the scene properties in different configurations becomes feasible. Among other applications, using the calibrated scene enables the accurate computation of the propagation channel response at different locations, through ray-tracing. The primary distinction between this differentiable ray-tracing method and the suggested approach is that the former emphasizes on calibrating a specific scene at the electromagnetic property level while considering the scene topology known. Conversely, the latter proposes a straightforward approach for the continuous location-to-channel mapping learning using a model-based neural network with no prior information about the scene topology.

    \textit{Organization.} The rest of the paper is organized as follows. Section~\ref{sec:problem_form} properly presents the location-to-channel mapping learning problem, Section~\ref{sec:expansion} the theoretical contributions. 
    Then, Section~\ref{sec:mb_archi} presents the translation of the proposed model into a neural architecture. Section~\ref{sec:experiments} proposes to evaluate the developed architecture against several baselines on realistic synthetic data. Finally, Section~\ref{sec:conclusion} introduces some conclusions and perspectives for future work.

    \textit{Notations.} Lowercase bold letters represent vectors while uppercase bold letters represent matrices. $\mathbb{R}$ and $\mathbb{C}$ respectively denotes the real and complex fields. $o$ denotes the small-$o$ Bachman-Landau notation. $\cdot ^\transp$ denotes the transpose matrix operator. $\text{diag}\left(\cdot\right)$ denotes the matrix operator constructing a diagonal matrix from a vector and $\text{vec}\left(\cdot\right)$ denotes the vectorization operator. $\mathbf{Id}_N \in \mathbb{R}^{N \times N}$ denotes the identity matrix, while $\mathbf{0}_N \in \mathbb{R}^{N \times N}$ denotes the null matrix. $\odot$ denotes the Hadamard product, and $\otimes$ denotes the Kronecker product. $\nabla_{\mathbf{a}} f\left(\mathbf{a},\mathbf{b}\right)$ denotes the gradient operator wrt. $\mathbf{a}$. $\abs{\cdot}$ denotes the absolute value for real numbers, modulus for complex numbers and cardinal for sets. $\norm{\cdot}{p}, p\in\mathbb{N}$ denotes the $\ell_p$ norm, and $\norm{\cdot}{F}$ denotes the $\ell_F$ (Frobenius) norm. $\mathcal{S}_1$ denotes the unit sphere while $\mathcal{C}_1$ denotes the unit circle. $\delta\left(t\right)$ denotes the Dirac impulse. $\mathfrak{Re} \left\{\cdot\right\}$ and $\mathfrak{Im} \left\{\cdot\right\}$ denote the real and imaginary parts. $\mathcal{U}$ denotes the uniform distribution.

\section{Problem formulation}\label{sec:problem_form}
    In this section, the physical propagation channel model is presented, and the mapping-learning problem is defined.

    Let us consider a communication system where a base station (BS) transmits information through $N_a$ antennas and $N_s$ distinct frequencies to mono-antenna user equipments (UEs). In the time domain, the propagation channel defines the filter operating on the electromagnetic waves transmitted between an emitting and receiving antenna. This filter impulse response is classically known as the channel impulse response (CIR). Considering the propagation channel over $L_p$ specular propagation paths between the $j$th BS antenna and a mono-antenna receiver located at $\mathbf{x} \in \mathbb{R}^3$ yields the following CIR definition:
    \begin{equation}
        \overline{h}_{j}\left(t,\mathbf{x}\right) = \sum_{l=1}^{L_p} \alpha_l\left(\mathbf{x}\right) \delta\left(t-\tau_l\left(\mathbf{x}\right)\right),\label{eq:model_time}
    \end{equation}
    where $\alpha_l\left(\mathbf{x}\right) \in \mathbb{R}$, resp. $\tau_l\left(\mathbf{x}\right) \in \mathbb{R}$, is the attenuation coefficient, resp. propagation delay, for the $l$th path.

    Applying the Fourier transform on the CIR yields the frequency channel response. For a given frequency $f_k$, this gives:
    \begin{align}
        h_{j}\left(f_k,\mathbf{x}\right) &=  \sum_{l=1}^{L_p} \alpha_l\left(\mathbf{x}\right)\mathrm{e}^{-\mathrm{j} 2\pi f_k \tau_l\left(\mathbf{x}\right)}\nonumber\\
        &= \sum_{l=1}^{L_p} \dfrac{\gamma_l}{d_l\left(\mathbf{x}\right)} \mathrm{e}^{-\mathrm{j}\frac{2\pi}{\lambda_k}d_l\left(\mathbf{x}\right)},\label{eq:model_freq}
    \end{align}
    where $\lambda_k \triangleq c/f_k \in \mathbb{R}$ is the wavelength associated to frequency $f_k$ and $d_l\left(\mathbf{x}\right) \in \mathbb{R}$ represents the propagation distance for the $l$th path. In further developments, the frequency dependence in Eq.~\eqref{eq:model_freq} is dropped and the notation $h_{j,k}\left(\mathbf{x}\right)$ is used. In Eq.~\eqref{eq:model_freq}, it is assumed that the attenuation is proportional to $d_l\left(\mathbf{x}\right)$ and that the delay is such that:
    \begin{equation}\label{eq:addit_delay}
        \tau_l\left(\mathbf{x}\right) \triangleq  \begin{cases}
            \frac{d_l\left(\mathbf{x}\right)}{c}, l=1\\
            \frac{d_l\left(\mathbf{x}\right)}{c} - \nu_l, l\neq 1
        \end{cases},
    \end{equation}
    where $\nu_l \in \mathbb{R}$ accounts for a supplementary delay induced by wave-matter interactions. As a result, the attenuation coefficient $\gamma_l$ becomes complex: $\forall l > 1$, $\mathfrak{Re}\left\{\gamma_l\right\}$ and $\mathfrak{Im}\left\{\gamma_l\right\}$ represent the small-scale attenuation and phase shift introduced along the $l$th path. Note that $\mathfrak{Re}\left\{\gamma_1\right\}=1$ and $\mathfrak{Im}\left\{\gamma_1\right\}=0$ when considering a Line of Sight (LoS) path, as this path does not present any wave-matter interaction.

    The contribution of reflected paths can be modeled using the virtual source theory~\cite[Chapter 1, pp.47-49]{pozar2011microwave}. This theory states that a path, between a source and a receiver, reflected on a plane $\mathcal{R}$ can be equivalently represented as a direct path originating from the source's symmetric counterpart with respect to $\mathcal{R}$, and ending at the receiver. This is illustrated in Fig.~\ref{fig:virtual_source_theory} (a). Moreover, while the virtual source theory is classically used to model reflections, it can also accommodate the modeling of diffraction~\cite{Deliang2014}. In this scenario, the virtual source is located at the endpoint of the diffraction object $\mathcal{D}$, as illustrated in Fig.~\ref{fig:virtual_source_theory} (b).

    \begin{figure}[t]
        \centering
        \includegraphics[width=.8\columnwidth]{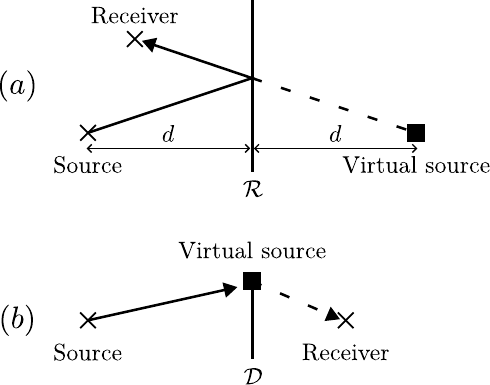}
        \caption{Virtual source theory used to model: (a) reflection, (b) diffraction.}
        \label{fig:virtual_source_theory}
    \end{figure}

    Using this theory to model the propagation distance $d_l\left(\mathbf{x}\right)$ yields:
    \begin{equation}
        h_{j,k}\left(\mathbf{x}\right) = \sum_{l=1}^{L_p} \dfrac{\gamma_l}{\norm{\mathbf{x}-\mathbf{a}_{l,j}}{2}} \mathrm{e}^{-\mathrm{j}\frac{2\pi}{\lambda_k}\norm{\mathbf{x}-\mathbf{a}_{l,j}}{2}},\label{eq:virtual_sources_model}
    \end{equation}
    where $\mathbf{a}_{l,j} \in \mathbb{R}^3$ denotes the location of the $l$th virtual antenna, so that $1/\norm{\mathbf{x}-\mathbf{a}_{l,j}}{2}$ represents the large scale fading of the $l$th path. 

    \begin{remark}
        Note that Eq.~\eqref{eq:virtual_sources_model} models a multipath channel, where the $L_p$ propagation paths are distinguishable paths: they can result from reflection, diffraction or scattering.
    \end{remark}

    \begin{remark}
            Note that the proposed channel model builds on the widely used ray-based channel model. This general framework subsumes specific models such as the Saleh-Valenzuela model, where path gains and ray directions follow specific statistical distributions. This general model provides greater flexibility to represent diverse channel scenarios, including indoor and outdoor environments, while also accounting for crucial physical propagation phenomena such as reflection and diffraction. This flexibility strengthens the general applicability of the proposed method.
    \end{remark}

    The exponential argument in Eq.~\eqref{eq:virtual_sources_model} reveals a high-frequency spatial dependence: a small change in the location $\mathbf{x}$ induces a significant change in the channel coefficient $h_{j,k}\left(\mathbf{x}\right)$. This arises from the wavelength dependence in the exponential argument: as the carrier frequency rises, the wavelength drops, yielding the fast variation of channel coefficients in the location space. In classical communication systems (sub-$6$ GHz), $\lambda$ is on the order of a few centimeters. 5G/6G systems also consider millimeter wavelengths.

    Let $\mathbf{H}\left(\mathbf{x}\right) \in \mathbb{C}^{N_a \times N_s}$ be the antenna-frequency channel matrix at location $\mathbf{x}$. This matrix is constructed from the concatenation of the complex frequency-channel coefficients $h_{j,k}\left(\mathbf{x}\right)$ in Eq.~\eqref{eq:virtual_sources_model}, across every antenna and frequency. The goal of this study is to learn:
    \begin{equation}
        \begin{aligned}
            f_{\boldsymbol{\theta}}\colon \mathbb{R}^3 &\longrightarrow \mathbb{C}^{N_a \times N_s}\\
            \mathbf{x} &\longrightarrow \mathbf{H}\left(\mathbf{x}\right),
        \end{aligned}
    \end{equation}
    a neural network $f$ parameterized by a set of learning parameters $\boldsymbol{\theta}$ that learns the continuous mapping between a location $\mathbf{x}$ and complex channel matrix $\mathbf{H}\left(\mathbf{x}\right)$. As mentioned in the introduction, the high-frequency spatial dependency in the channel model causes the location-to-channel mapping to be remarkably hard to learn using classical neural architectures, due to the spectral-bias issue. 
    Alternative neural architectures overcoming this issue are presented in the next section.

\section{Towards a model-based approach for the location-to-channel mapping learning problem}\label{sec:expansion}
    This section presents the proposed model-based approach, as well as the related theoretical results. Such theoretical developments are crucial for understanding the proposed model-based neural architecture presented in Section~\ref{sec:mb_archi}, as this architecture is based on Theorem~\ref{theorem:global_validity_approx}. Although the use of Taylor approximations on the propagation distance is not novel, this paper is, to the best of the authors' knowledge, the first to propose an approximated channel model that enables extrapolation around both a reference receiver location and a reference emitter antenna location, relying solely on assumptions of attenuation and phase proportionality with the propagation distance. Additionally, this section presents a discussion on the use of sparse recovery methods to address the problem formulated in Eq.~\eqref{eq:global_approx}.

    \subsection{Model-based approach}\label{subsection:mb_approach}
    In this paper, it is proposed to use knowledge from the physics-based propagation model presented in Eq.~\eqref{eq:virtual_sources_model} to derive a model-based neural architecture that would overcome the spectral-bias issue. As the high-frequency spatial dependency originates from $\norm{\mathbf{x}-\mathbf{a}_{l,j}}{2}/\lambda_k$ in Eq.~\eqref{eq:virtual_sources_model}, it is proposed to develop this term using a Taylor expansion. It will be shown that it allows to separate the high-frequency from the low-frequency spatial content. This approach, known as the plane wave approximation, is used to obtain the well-known steering vector model~\cite[Chapter 7]{Tse2005},~\cite{Lemagoarou2019}. 
    \begin{lemma}\label{lemma:taylor_approx}
        Let $\mathbf{x}_r \in \mathbb{R}^3$ be a reference location and $\mathcal{D}_{\mathbf{x}} \subset \mathbb{R}^3$ be a local validity domain such that $\forall \mathbf{x} \in \mathcal{D}_{\mathbf{x}}, \norm{\mathbf{x}-\mathbf{x}_r}{2} \leq \epsilon_{\mathbf{x}}$. Let $\mathbf{a}_{l,r} \in \mathbb{R}^3$ be a reference antenna location and $\mathcal{D}_{\mathbf{a}} \subset \mathbb{R}^3$ be local validity domain such that $\forall \mathbf{a}_{l,j} \in \mathcal{D}_{\mathbf{a}}, \norm{\mathbf{a}_{l,j}-\mathbf{a}_{l,r}}{2} \leq \epsilon_{\mathbf{a}}$. One has, $\forall \left(\mathbf{x},\mathbf{a}_{l,j}\right) \in \mathcal{D}_{\mathbf{x}} \times \mathcal{D}_{\mathbf{a}}:$ 
        \begin{align}\label{eq:norm_approx}
            \norm{\mathbf{x}-\mathbf{a}_{l,j}}{2} \simeq &\norm{\mathbf{x}_r - \mathbf{a}_{l,r}}{2} + \mathbf{u}_{l,j}\left(\mathbf{x}_r\right)^\transp\left(\mathbf{x}-\mathbf{x}_r\right)\\
            &- \mathbf{u}_{l,r}\left(\mathbf{x}_r\right)^\transp \left(\mathbf{a}_{l,j}-\mathbf{a}_{l,r}\right), \nonumber
        \end{align}
        with $\mathbf{u}_{l,j}\left(\mathbf{x}_r\right) = \left(\mathbf{x}_r - \mathbf{a}_{l,j}\right)/\norm{\mathbf{x}_r - \mathbf{a}_{l,j}}{2}$.
    \end{lemma}
    \begin{proof}
        See Appendix~\ref{apendix:taylor}.
    \end{proof}

    \begin{figure}[t]
        \centering
        \includegraphics[scale=.8]{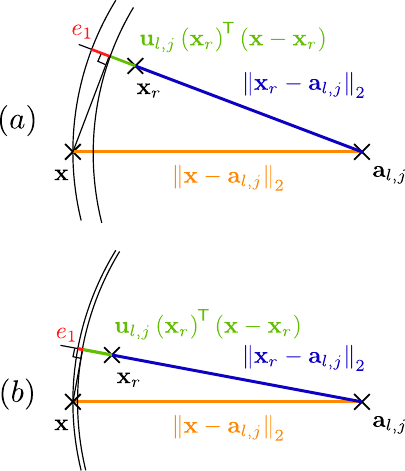}
        \caption{Taylor expansion on locations only: $(a)$ location $\mathbf{x}$ far from reference $\mathbf{x}_r$, $(b)$ location $\mathbf{x}$ close to reference $\mathbf{x}_r$.}
        \label{fig:taylor_loc_schematic}
    \end{figure}

    Fig.~\ref{fig:taylor_loc_schematic} illustrates the Taylor expansion on locations. One can remark the direct link between the approximation error and the distance to reference: when the considered location is far from the reference, the approximation error is high, while on the other hand, when the distance to the reference is small, the approximation error decreases. This emphasizes the local validity of Lemma~\ref{lemma:taylor_approx}. A similar analysis can be made for the Taylor expansion on antennas.

    Corollary~\ref{corollary:error_approx} proposes to analytically characterize the approximation error in Lemma~\ref{lemma:taylor_approx}.
    \begin{corollary}\label{corollary:error_approx}
        Let $\mathbf{x}_r \in \mathbb{R}^3$ and $\mathbf{a}_{l,r} \in \mathbb{R}^3$ be a reference location and a reference antenna location respectively. The approximation error in Eq.~\eqref{eq:norm_approx} can be approximated using the second order of the Taylor expansions as:
        \begin{align}
            e \simeq \dfrac{1}{2}&\left(\dfrac{\norm{\mathbf{x}-\mathbf{x}_r}{2}^2}{\norm{\mathbf{x}_r-\mathbf{a}_{l,j}}{2}}  - o\left(\dfrac{1}{\norm{\mathbf{x}_r-\mathbf{a}_{l,j}}{2}^2}\right)\right.\\
            &\left. + \dfrac{\norm{\mathbf{a}_{l,j}-\mathbf{a}_{l,r}}{2}^2}{\norm{\mathbf{x}_r-\mathbf{a}_{l,r}}{2}}  - o\left(\dfrac{1}{\norm{\mathbf{x}_r-\mathbf{a}_{l,r}}{2}^2}\right) \right) \nonumber
        \end{align}
    \end{corollary}
    \begin{proof}
        See Appendix~\ref{appendix:cor_1_proof}.
    \end{proof}

    One can remark in Corollary~\ref{corollary:error_approx} that the approximation error is mostly dependent on the ratios $\norm{\mathbf{x}-\mathbf{x}_r}{2}^2/\norm{\mathbf{x}_r-\mathbf{a}_{l,j}}{2}$ and $\norm{\mathbf{a}_{l,j}-\mathbf{a}_{l,r}}{2}^2/\norm{\mathbf{x}_r-\mathbf{a}_{l,r}}{2}$. In other words, when the distance between the reference location and current/reference antenna is important, the approximation error decreases. Indeed, in such scenario, the common far field approximation is valid, making the planar approximation of spherical wavefronts realistic. In Eq.~\eqref{eq:norm_approx}, the planar wavefront term is represented by the first line projection term: this term is null when the location $\left(\mathbf{x}-\mathbf{x}_r\right)$ is orthogonal to the direction $\mathbf{u}_{l,j}\left(\mathbf{x}_r\right)$. Once introduced in a complex exponential, it yields periodic parallel level sets in the $\mathbf{u}_{l,j}\left(\mathbf{x}_r\right)$ direction. Such phenomenon constitutes the definition of planar wavefronts. On the other hand, when the reference distances are small, the far field approximation does not hold, making the approximation error being directly dependent on the distances to the references $\norm{\mathbf{x}-\mathbf{x}_r}{2}$ and $\norm{\mathbf{a}_{l,j}-\mathbf{a}_{l,r}}{2}$.

    Proposition~\ref{proposition:channel_approx} presents the injection of the approximated propagation distance in the channel model in Eq.~\eqref{eq:virtual_sources_model}. This allows to obtain an approximated channel coefficient around a reference location and a reference antenna.

    \begin{mdframed}
        \begin{proposition}\label{proposition:channel_approx}
        Let $\mathbf{x}_r \in \mathbb{R}^3$ and $\mathbf{a}_{l,r} \in \mathbb{R}^3$ be a reference location and reference antenna location. Let $d_{l,r} \triangleq \norm{\mathbf{x}_r-\mathbf{a}_{l,r}}{2}$, $\tau_{l,r} \triangleq d_{l,r}/c$, and $h_{l,r}\left(\mathbf{x}_r\right) \triangleq \mathrm{e}^{-\mathrm{j}\frac{2\pi}{\lambda_r}d_{l,r}}/d_{l,r}$. Let $f_r \in \mathbb{R}$ be a reference frequency such that $\forall f_k \in \mathbb{R}, f_k = \left(f_k - f_r\right)+f_r$. One has, $\forall \left(\mathbf{x},\mathbf{a}_{l,j}\right) \in \mathcal{D}_{\mathbf{x}} \times \mathcal{D}_{\mathbf{a}}:$
        \begin{align}\label{eq:chan_approx}
            h_{j,k}\left(\mathbf{x}\right) \simeq \sum_{l=1}^{L_p} &\underbrace{\gamma_l h_{l,r}\left(\mathbf{x}_r\right)}_{\normalfont \text{Reference channel}} \underbrace{\vphantom{\gamma_l h_{l,r}\left(\mathbf{x}_r\right)} \mathrm{e}^{-\mathrm{j}\frac{2\pi}{\lambda_{r}}\mathbf{u}_{l,r}\left(\mathbf{x}_r\right)^\transp\left(\mathbf{x}-\mathbf{x}_r\right)}}_{\normalfont \text{Location correction}}\\
            &\cdot \underbrace{\mathrm{e}^{-\mathrm{j}2\pi \left(f_k-f_r\right)\tau_{l,r}}}_{\normalfont \text{Frequency correction}} \underbrace{\mathrm{e}^{\mathrm{j}\frac{2\pi}{\lambda_{r}}\mathbf{u}_{l,r}\left(\mathbf{x}_r\right)^\transp\left(\mathbf{a}_{l,j}-\mathbf{a}_{l,r}\right)}}_{\normalfont \text{Antenna correction}}.\nonumber
        \end{align}
    \end{proposition}
    \end{mdframed}
    \begin{proof}
        See Appendix~\ref{appendix:c}.
    \end{proof}

    It can be seen in Eq.~\eqref{eq:chan_approx} that $h_{j,k}\left(\mathbf{x}\right)$ is a sum of reference channels around a reference location at a reference frequency multiplied by location, frequency, and antenna correction terms. Letting $w_{l,r}\left(\mathbf{x}\right) \triangleq \gamma_l h_{l,r}\left(\mathbf{x}_r\right)\mathrm{e}^{\mathrm{j}\frac{2\pi}{\lambda_r}\mathbf{u}_{l,r}\left(\mathbf{x}_r\right)^\transp\mathbf{x}_r} \in \mathbb{C}$, Eq.~\eqref{eq:chan_approx} can be further rearranged as, $\forall \left(\mathbf{x},\mathbf{a}_{l,j}\right) \in \mathcal{D}_{\mathbf{x}} \times \mathcal{D}_{\mathbf{a}}$:
    \begin{align}\label{eq:spectral_separation}
        h_{j,k}\left(\mathbf{x}\right) \simeq \sum_{l=1}^{L_p}& \underbrace{w_{l,r}\left(\mathbf{x}\right) \mathrm{e}^{-\mathrm{j}2\pi\left(f_k-f_r\right) \tau_{l,r}} \mathrm{e}^{\mathrm{j}\frac{2\pi}{\lambda_r}\mathbf{u}_{l,r}\left(\mathbf{x}_r\right)^\transp \left(\mathbf{a}_{l,j}-\mathbf{a}_{l,r}\right)}}_{\normalfont \text{Slowly varying}} \nonumber\\
        &\cdot \underbrace{\mathrm{e}^{-\mathrm{j}\frac{2\pi}{\lambda_r}\mathbf{u}_{l,r}\left(\mathbf{x}_r\right)^\transp\mathbf{x}}}_{\normalfont \text{Fastly varying}}. 
    \end{align}

    \begin{remark}
        One can observe in Eq.~\eqref{eq:spectral_separation} that the channel coefficient $h_{j,k}\left(\mathbf{x}\right)$ can be viewed as a linear combination of planar wavefronts. The directions of said planar wavefronts are defined by the spatial frequencies $\mathbf{u}_{l,r}\left(\mathbf{x}_r\right)$. Furthermore, Eq.~\eqref{eq:spectral_separation} presents a spectral separation stage: the planar wavefronts present high-frequency spatial content because of their wavelength dependency in the exponential argument, while the weights, multiplied by the location and antenna-correction terms, present low-frequency spatial content (order of the Taylor-expansions validity domain).
    \end{remark}

    Proposition~\ref{proposition:matrix_approximation} presents the expansion of the obtained approximation to every antenna and frequency, i.e. expresses the channel matrix $\mathbf{H}\left(\mathbf{x}\right)$ as a function of the location $\mathbf{x}$.
    \begin{mdframed}
    \begin{proposition}\label{proposition:matrix_approximation}
        Let $\boldsymbol{\psi}_{\mathbf{a},l}\left(\mathbf{x}\right) \in \mathbb{C}^{N_a}$ be a steering vector (SV) and $\boldsymbol{\psi}_{\mathbf{f},l}\left(\mathbf{x}\right) \in \mathbb{C}^{N_s}$ be a frequency response vector (FRV). Let $\psi_{\mathbf{x},l}\left(\mathbf{x}\right) \triangleq \mathrm{e}^{-\mathrm{j}\frac{2\pi}{\lambda_{r}}\mathbf{u}_{l,r}\left(\mathbf{x}_r\right)^\transp \mathbf{x}} \in \mathbb{C}$ be a planar wavefront. The multi-antenna multi-frequency channel can be expressed as, $\forall \left(\mathbf{x},\mathbf{a}_{l,j}\right) \in \mathcal{D}_{\mathbf{x}} \times \mathcal{D}_{\mathbf{a}}:$
        \begin{align}\label{eq:matrix_chan_approx}
            \mathbf{H}\left(\mathbf{x}\right) \simeq \sum_{l=1}^{L_p}& w_{l,r}\left(\mathbf{x}\right) \psi_{\mathbf{x},l}\left(\mathbf{x}\right) \boldsymbol{\psi}_{\mathbf{a},l}\left(\mathbf{x}_r\right) \boldsymbol{\psi}_{\mathbf{f},l}\left(\mathbf{x}_r\right)^\transp,
        \end{align}
        where
        \begin{equation}\label{eq:ant_correction}
            \boldsymbol{\psi}_{\mathbf{a},l}\left(\mathbf{x}_r\right) = \begin{bmatrix}
                \mathrm{e}^{\mathrm{j}\frac{2\pi}{\lambda_{r}}\mathbf{u}_{l,r}\left(\mathbf{x}_r\right)^\transp\left(\mathbf{a}_{l,1}-\mathbf{a}_{l,r}\right)}\\
                \vdots\\
                \mathrm{e}^{\mathrm{j}\frac{2\pi}{\lambda_{r}}\mathbf{u}_{l,r}\left(\mathbf{x}_r\right)^\transp\left(\mathbf{a}_{l,N_a}-\mathbf{a}_{l,r}\right)}
            \end{bmatrix} \in \mathbb{C}^{N_a},
        \end{equation}
        and
        \begin{equation}\label{eq:freq_correction}
            \boldsymbol{\psi}_{\mathbf{f},l}\left(\mathbf{x}_r\right) = \begin{bmatrix}
                \mathrm{e}^{-\mathrm{j}2\pi \left(f_1-f_r\right)\tau_{l,r}}\\
                \vdots\\
                \mathrm{e}^{-\mathrm{j}2\pi \left(f_{N_s}-f_r\right)\tau_{l,r}}
            \end{bmatrix}\in \mathbb{C}^{N_s}.
        \end{equation}
    \end{proposition}
    \end{mdframed}
    \begin{proof}
        Direct derivation from Eq.~\eqref{eq:spectral_separation} by defining SVs and FRVs.
    \end{proof}
    The previously obtained results can be summarized as follows:
    \begin{itemize}
        \item Eq.~\eqref{eq:norm_approx} presents the Taylor expansion of the propagation distance around a reference location and reference antenna location.
        \item Eq.~\eqref{eq:spectral_separation} presents how the Taylor expansions introduce a spectral separation stage in the propagation channel.
        \item Eq.~\eqref{eq:matrix_chan_approx} presents the expansion of the previous results over every antenna and frequencies.
    \end{itemize}

    One can remark that the approximation proposed in Eq.~\eqref{eq:matrix_chan_approx} is only valid in local neighborhoods of the reference antenna location $\mathbf{a}_{l,r}$ and reference location $\mathbf{x}_r$, namely $\mathcal{D}_{\mathbf{a}}$ and $\mathcal{D}_{\mathbf{x}}$. Theorem~\ref{theorem:global_validity_approx} presents the expansion of the approximation validity domain to the entire $\mathbb{R}^3$ space. The idea is to partition $\mathbb{R}^3$ into location and antenna local validity domains, and then aggregate the needed planar wavefronts, SVs, and FRVs into dictionaries for each location and antenna local validity domain pair. Additionally, proof of Theorem~\ref{theorem:global_validity_approx} introduces the Direction of Departure (DoD) $\tilde{\mathbf{u}}_i$ so that the SV dictionary is constructed with only the physical antenna locations $\mathbf{a}_{1,j}$.
    \begin{mdframed}
    \begin{theorem}\label{theorem:global_validity_approx}
        Let us consider the tiling of the location subset $\mathcal{S}_{\mathbf{x}} \subset \mathbb{R}^3$ into $\Omega_{\mathbf{x}}$ local validity domains. A second tiling is applied for the antenna subset $\mathcal{S}_{\mathbf{a}} \subset \mathbb{R}^3$ with $\Omega_{\mathbf{a}}$ local validity domains. Let $D \in \mathbb{N}^* \text{ st. } D \leq L_p \Omega_{\mathbf{x}} \Omega_{\mathbf{a}}$. Let $\tilde{\mathbf{\Psi}}_{\mathbf{a}} \in \mathbb{C}^{N_a \times D}$ be a dictionary of SVs, defined as:
        \begin{equation}\label{eq:sv_dict_def}
            \tilde{\mathbf{\Psi}}_{\mathbf{a}} = \left\{\tilde{\boldsymbol{\psi}}_{\mathbf{a},i}\right\}_{i=1}^D = \left\{ \begin{bmatrix}
                \mathrm{e}^{\mathrm{j}\frac{2\pi}{\lambda_{r}}\tilde{\mathbf{u}}_{i}^\transp\left(\mathbf{a}_{1,1}-\mathbf{a}_{1,r}\right)}\\
                \vdots\\
                \mathrm{e}^{\mathrm{j}\frac{2\pi}{\lambda_{r}}\tilde{\mathbf{u}}_{i}^\transp\left(\mathbf{a}_{1,N_a}-\mathbf{a}_{1,r}\right)}
            \end{bmatrix} \right\}_{i=1}^D,
        \end{equation}
        where $\tilde{\mathbf{u}}_{i} \in \mathbb{R}^3$ is a DoD. Let $\tilde{\mathbf{\Psi}}_{\mathbf{f}} \in \mathbb{C}^{N_s \times D}$ be a dictionary of FRVs, defined as:
        \begin{equation}\label{eq:frv_dict_def}
            \tilde{\mathbf{\Psi}}_{\mathbf{f}} = \left\{\tilde{\boldsymbol{\psi}}_{\mathbf{f},i}\right\}_{i=1}^D = \left\{\begin{bmatrix}
                \mathrm{e}^{-\mathrm{j}2\pi \left(f_1-f_r\right)\tau_{i}}\\
                \vdots\\
                \mathrm{e}^{-\mathrm{j}2\pi \left(f_{N_s}-f_r\right)\tau_{i}}
            \end{bmatrix}\right\}_{i=1}^D,
        \end{equation} 
        where $\tau_i \in \mathbb{R}^{+}$ is a propagation delay. Let $\tilde{\boldsymbol{\psi}}_{\mathbf{x}}\left(\mathbf{x}\right)\in \mathbb{C}^D$ be a dictionary of planar wavefronts, defined as:
        \begin{equation}\label{eq:plan_wave_dict}
            \tilde{\boldsymbol{\psi}}_{\mathbf{x}}\left(\mathbf{x}\right) =  \left\{\tilde{\psi}_{\mathbf{x},i}\right\}_{i=1}^D = \left\{\mathrm{e}^{-\mathrm{j}\frac{2\pi}{\lambda_r}\mathbf{u}_i^\transp \mathbf{x}}\right\}_{i=1}^D,
        \end{equation}
        where $\mathbf{u}_i \in \mathbb{R}^3$ a spatial frequency. Finally,
        let $\mathbf{w}\left(\mathbf{x}\right) \in \mathbb{C}^D$ be an activation vector such that: $\boldsymbol{\varpi}\left(\mathbf{x}\right) = \mathbf{w}\left(\mathbf{x}\right) \odot  \tilde{\boldsymbol{\psi}}_{\mathbf{x}}\left(\mathbf{x}\right)$. Then, $\forall \mathbf{x} \in \mathbb{R}^3$:
        \begin{align}\label{eq:global_approx}
            &\mathbf{H}\left(\mathbf{x}\right)\simeq \sum_{i=1}^{D} \varpi_i\left(\mathbf{x}\right) \tilde{\boldsymbol{\psi}}_{\mathbf{a},i} \tilde{\boldsymbol{\psi}}_{\mathbf{f},i}^\transp \\
            &\text{with} \norm{\boldsymbol{\varpi}\left(\mathbf{x}\right)}{0} = L_p, \nonumber 
        \end{align}
    \end{theorem}
    \end{mdframed}
    \begin{proof}
        See Appendix~\ref{appendix:d}.
    \end{proof}

    One can see that Theorem~\ref{theorem:global_validity_approx} proposes a sparse continuous interpolation of the channel matrix. The composite dictionary consists of all possible combinations of the required SVs and FRVs, while the activation vector presents spectral separation: it is defined with a low-frequency term $\mathbf{w}\left(\mathbf{x}\right)$, and a high-frequency term $\tilde{\boldsymbol{\psi}}_{\mathbf{x}}\left(\mathbf{x}\right)$.

    \begin{remark}
        While Theorem~\ref{theorem:global_validity_approx} extends the validity expansion of Eq.~\eqref{eq:matrix_chan_approx} to the entire $\mathbb{R}^3$ space, the number of atoms $D$ in each dictionary is intractable due to its dependence on the number of local validity domains $ \Omega_{\mathbf{x}}$ and $ \Omega_{\mathbf{a}}$. A way to overcome this issue is to construct each dictionary by discretizing its generating subspace, i.e. the DoD subspace, unit sphere $\mathcal{S}_1$, for the SV dictionary, the delay subspace $\mathbb{R}^+$ for the FRV dictionary and the spatial frequency subspace $\mathcal{S}_1$ for the planar wavefront dictionary. This approach is at the center of the model-based neural network architecture presented in Section~\ref{sec:mb_archi}.
    \end{remark}
    \subsection{Discussion: a sparse recovery approach}
    As Eq.~\eqref{eq:global_approx} can be viewed as a sparse reconstruction problem, one could think of finding the correct activation coefficients using sparse recovery techniques. Indeed, using matching pursuit (MP)/orthogonal matching pursuit (OMP) algorithms~\cite{Mallat1993, Pati1993}, one could find a sparse representation of $\mathbf{H}\left(\mathbf{x}\right)$ in a composite dictionary of FRVs/SVs. Obtaining a channel estimate $\hat{\mathbf{H}}\left(\mathbf{x}\right)$ at a wanted location $\mathbf{x}$ could be separated into two stages. Firstly, one would have to solve:
    \begin{mini}|s|
        {\mathbf{w}}{\norm{\mathbf{H}\left(\mathbf{x}_r\right)-\sum_{i=1}^{D} w_i\tilde{\psi}_{\mathbf{x},i}\left(\mathbf{x}_r\right) \tilde{\boldsymbol{\psi}}_{\mathbf{a},i} \tilde{\boldsymbol{\psi}}_{\mathbf{f},i}^\transp}{F}^2}{}{},
        \addConstraint{\norm{\mathbf{w}}{0} = L_p},
        \label{eq:sparse_prob}
    \end{mini}
    $\forall \mathbf{x}_r \in \mathcal{H}_r \subset \mathbb{R}^3$, $\mathcal{H}_r$ being a set of reference locations. Then, for the wanted location $\mathbf{x}$, one could obtain the channel estimate as:
    \begin{equation}
        \hat{\mathbf{H}}\left(\mathbf{x}\right) = \sum_{i=1}^{D} w_{\mathbf{x}^\star_{r,i}} \tilde{\psi}_{\mathbf{x},i}\left(\mathbf{x}\right) \tilde{\boldsymbol{\psi}}_{\mathbf{a},i} \tilde{\boldsymbol{\psi}}_{\mathbf{f},i}^\transp,
    \end{equation}
    where $\mathbf{w}_{\mathbf{x}^\star_r} \in \mathbb{C}$ is the activation coefficient obtained by solving Eq.~\eqref{eq:sparse_prob} for $\mathbf{x}^\star_r = \argmin_{\mathbf{x}_r \in \mathcal{H}_r} \norm{\mathbf{x}-\mathbf{x}_r}{2}$, i.e. the closest reference location to the wanted location.

    However this approach presents an intractable complexity induced by the high-frequency nature of the propagation channel and the composite dictionary size. Indeed, as the propagation channel is a function rapidly varying with the location, one would have to consider spatially close reference locations, i.e. with spacing between reference locations under the wavelength. This introduces a high cardinality in $\mathcal{H}_r$, which increases the time complexity for achieving the first step of this estimation scheme. Additionally, as the composite dictionary is composed of every pair of FRVs/SVs, it results in a prohibitive dictionary size, further increasing the time complexity of the first estimation step.

    The next section proposes to transform the obtained channel approximation in Eq.~\eqref{eq:global_approx} into a neural architecture, following the model-based machine learning paradigm.

\section{Model-based neural architecture}\label{sec:mb_archi}
    As it has been shown that classical signal processing methods are unsuited for solving the sparse reconstruction problem in Theorem~\ref{theorem:global_validity_approx}, it is proposed to leverage the learning capabilities of neural networks for the location-to-channel mapping learning. This section presents a model-based neural architecture from Eq.~\eqref{eq:global_approx}.

    The approximation in Eq.~\eqref{eq:global_approx} can be rewritten using the vectorization operator, clearly displaying the sparse representation structure of the approximated channel. $\forall \mathbf{x}\in \mathbb{R}^3$:
    \begin{align}\label{eq:global_matrix_form}
        &\text{vec}\left(\mathbf{H}\left(\mathbf{x}\right)\right)\simeq \left(\tilde{\mathbf{\Psi}}_{\mathbf{f}}\left(\mathbf{x}\right) \otimes \tilde{\mathbf{\Psi}}_{\mathbf{a}}\left(\mathbf{x}\right)  \right) \text{vec}\left(\text{diag}\left(\boldsymbol{\varpi}\left(\mathbf{x}\right)\right)\right) \\
        &\text{with} \norm{\boldsymbol{\varpi}\left(\mathbf{x}\right)}{0} = L_p, \nonumber 
    \end{align}

    As shown in Eq.~\eqref{eq:global_matrix_form}, the location-to-channel mapping learning problem can be partitioned into four subproblems: the planar wavefronts, FRV, SV dictionaries, and the complex weights coefficients learning.

    \begin{figure}[!h]
        \centering
        \includegraphics[scale=.4]{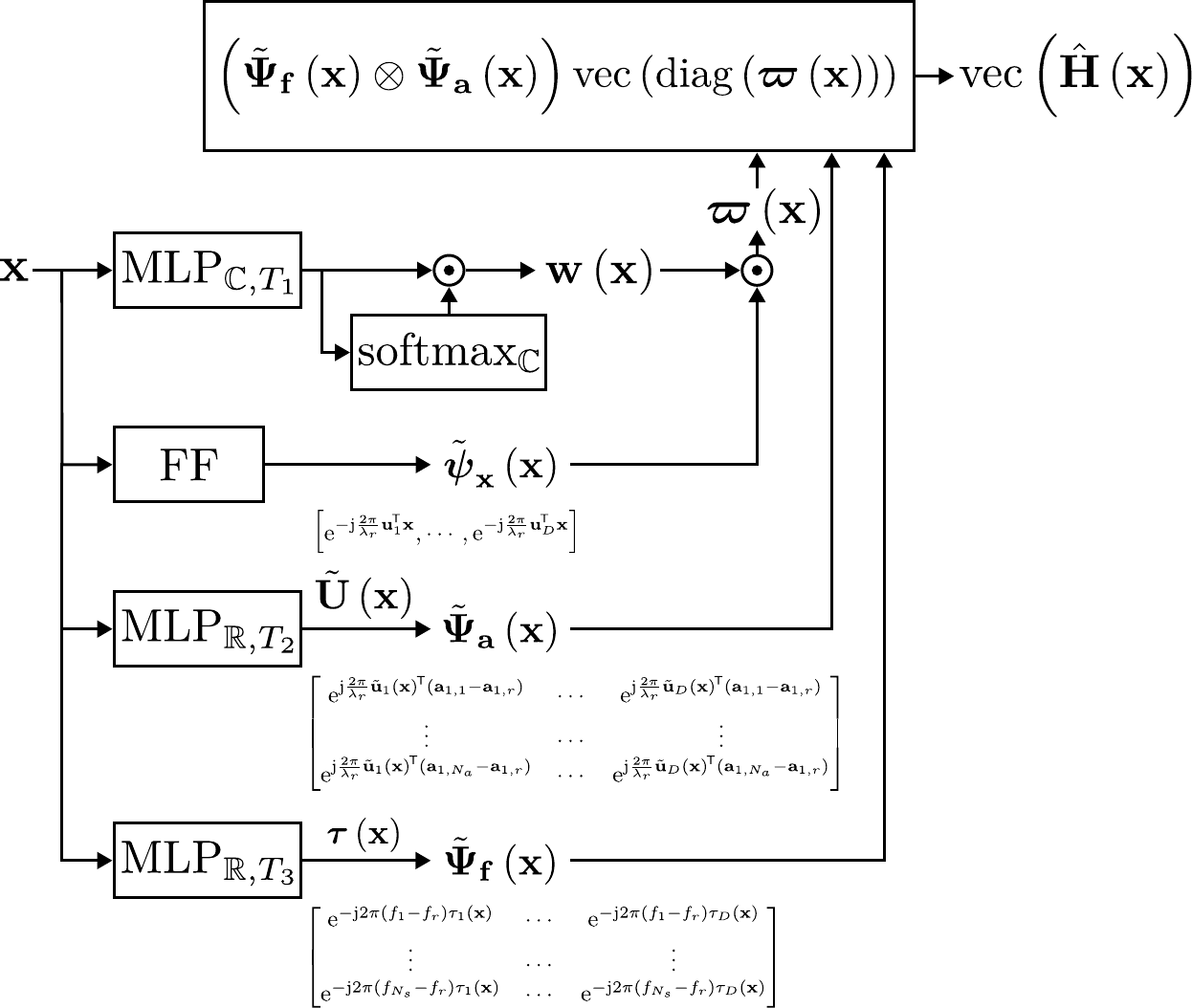}
        \caption{Proposed model-based neural architecture.}
        \label{fig:mb_neural_archi}
    \end{figure}

    \subsection{Planar wavefront dictionary}
    As exposed by the previous analysis, the planar wavefronts take the form $\tilde{\boldsymbol{\psi}}_{\mathbf{x}}\left(\mathbf{x}\right) = \{\mathrm{e}^{-\mathrm{j}\frac{2\pi}{\lambda_r}\mathbf{u}_i^\transp \mathbf{x}}\}_{i=1}^D \in \mathbb{C}^D$. Such dictionary can be constructed from the spatial frequencies $\mathbf{u}_i \in \mathbb{R}^3$, which by definition, are unit vectors representing the planar wavefronts direction. As a result, every spatial frequency belongs to a two-dimensional manifold: the unit sphere $\mathcal{S}_1$. Sampling $\mathcal{S}_1$ with $D$ points yields the spatial frequency dictionary $\mathbf{U} \triangleq \{\mathbf{u}_i\}_{i=1}^D \in \mathbb{R}^{3\times D}$ that is used to compute the planar wavefront dictionary. The dictionary $\tilde{\boldsymbol{\psi}}_{\mathbf{x}}\left(\mathbf{x}\right)$ is then constructed using the Fourier feature (FF) layer, defined as:
    \begin{equation}
        \text{FF}: \mathbf{x} \rightarrow \left[ \mathrm{e}^{-\mathrm{j}\frac{2\pi}{\lambda_r}\mathbf{u}_1^\transp \mathbf{x}}, \cdots, \mathrm{e}^{-\mathrm{j}\frac{2\pi}{\lambda_r}\mathbf{u}_D^\transp \mathbf{x}} \right].
    \end{equation}
    Note that this FF layer is just a complex implementation of the $\cos$/$\sin$ embedding layer used in RFF networks~\cite{rahimi2007,tancik2020fourfeat}. Also note that spatial frequencies could be learned either by gradient descent or directly through a neural network, however, they are kept fixed for this study.


    \subsection{FRV dictionary}
    It has been presented in Eq.~\eqref{eq:frv_dict_def} that the FRV dictionary $\tilde{\mathbf{\Psi}}_{\mathbf{f}}\left(\mathbf{x}\right) \in \mathbb{C}^{N_s \times D}$ only depends on the system frequencies $f_k$, the reference frequency $f_r$, and propagation delays $\tau_i$. The system and reference frequencies are assumed to be known: such assumption is typically made in classical communication systems. In this study, the reference frequency is computed as $f_r = \frac{1}{N_s}\sum_{k=1}^{N_s} f_k$. The FRV dictionary can then be constructed by sampling the propagation delay subspace, i.e. $\mathbb{R}^+$. In order to optimize the performance of the proposed approach, it is suggested to learn the discretization for every location. As such, it is proposed to use a MLP that learns a propagation delay vector $\boldsymbol{\tau}\left(\mathbf{x}\right) \in \mathbb{R}^D$ for every location $\mathbf{x}$. Indeed, as the propagation delay vector exhibits slow variations in the location space, the MLP will not suffer from the spectral bias.

    \subsection{SV dictionary}
    As presented in Eq.~\eqref{eq:sv_dict_def}, the SV dictionary $\tilde{\mathbf{\Psi}}_{\mathbf{a}}\left(\mathbf{x}\right) \in \mathbb{C}^{N_a \times D}$ only depends on the reference frequency wavelength $\lambda_r$, on the true antenna locations $\mathbf{a}_{1,j}$, on the true reference antenna location $\mathbf{a}_{1,r}$, and on DoDs $\tilde{\mathbf{u}}_i$. The reference frequency wavelength is assumed to be known, as well as the antenna locations. The true reference antenna location is computed as the barycenter of the true antenna locations. In order to construct the SV dictionary, one could discretize the DoD subspace $\mathcal{S}_1$. Similarly than for the FRV dictionary, it is proposed to optimize performance by learning the discretization for every location. As such, it is proposed to learn $\tilde{\mathbf{U}}\left(\mathbf{x}\right)$ using a MLP, a method referred to as MB-$\tilde{\mathbf{u}}$ learning. On the other hand, it is proposed to directly learn the entire SV dictionary for each location using a MLP, without assuming any dictionary structure. This approach, referred to as MB-$\tilde{\mathbf{\Psi}}_{\mathbf{a}}$ learning releases constraints on the antenna correction terms at the expense of an increased learning-parameter complexity.

    \subsection{Complex weights learning}\label{subsection:complex_weights}
    The complex activation weights $\mathbf{w}\left(\mathbf{x}\right) \in \mathbb{C}^D$ are learned following the same approach as in~\cite{chatelier2023modelbased}. The MLP neural network is used to learn the complex weight vector, $\mathbf{w}\left(\mathbf{x}\right) \in \mathbb{C}^D$, introducing the sparsity constraint with a \texttt{softmax} ponderation. Note that this hypernetwork encompass the learning of the path attenuation: while the channel model presented in Eq.~\eqref{eq:model_freq} assumes the attenuation being proportional to the propagation distance $d_l\left(\mathbf{x}\right)$, there are cases where attenuation is proportional to $d_l\left(\mathbf{x}\right)^\alpha$, with $\alpha \in \mathbb{R}^{+,*}$. The proposed neural model can handle such cases, as the MLP neural network used for the complex activation weights learning is capable to learn such non-linear content. Finally, this MLP can adapt the sparsity of the learned activation vector, thereby effectively accounting for potential variations in the number of propagation paths for different locations within the considered scene.

    \subsection{Global architecture}
    The proposed model-based neural network architecture for the location-to-channel mapping learning is presented in Fig.~\ref{fig:mb_neural_archi}. Note that this architecture only presents the MB-$\tilde{\mathbf{u}}$ approach. During experiments, the used approach will be explicitly specified. Also note that $\text{MLP}_{\mathbb{C},T}$ represents a 3-layer $\texttt{ReLU}_{\mathbb{C}}$-MLP~\footnote{Note that $\forall z_1 \in \mathbb{C}, \forall \mathbf{z}_2 \in \mathbb{C}^N$, $\texttt{ReLU}_{\mathbb{C}}\left(z_1\right) = \texttt{ReLU}\left(\mathfrak{Re} \left\{z_1\right\}\right) + \mathrm{j} \texttt{ReLU}\left(\mathfrak{Im} \left\{z_1\right\}\right)$ and $\texttt{softmax}_\mathbb{C}\left(\mathbf{z}_2\right) = \texttt{softmax}\left(\abs{\mathbf{z}_2}\right)$.} with complex weights and biases, where each layer width is $T$. $\text{MLP}_{\mathbb{R},T}$ represents the same MLP but with real weights/biases and $\texttt{ReLU}$ activation functions.
    
    The following remark emphasizes that the use of hypernetworks is a key feature of the proposed model.

    \begin{remark}
        The proposed architecture shares common features with classical INR networks. Indeed, the FF layer can be seen as a positional-embedding stage, where the low dimensional location information is projected on a higher dimensional space containing high spatial frequencies. Nevertheless, the proposed architecture stands out from classical INR networks with the introduction of hypernetworks~\cite{Schmidhuber1992,Ha2017}, i.e. parallel neural networks learning parameters of the main network. As illustrated in Fig.~\ref{fig:mb_neural_archi}, four hypernetworks are used to learn parameters of the main branch: one for the complex activation vector learning $\mathbf{w}\left(\mathbf{x}\right)$, and three dedicated for the planar wavefront $\tilde{\boldsymbol{\psi}}_{\mathbf{x}}\left(\mathbf{x}\right)$, FRV $\tilde{\mathbf{\Psi}}_{\mathbf{f}}\left(\mathbf{x}\right)$ and SV $\tilde{\mathbf{\Psi}}_{\mathbf{a}}\left(\mathbf{x}\right)$ dictionaries constitution. Note that the proposed architecture is a specific instance of the hypernetwork architecture family, which includes the widely popular self-attention mechanism.
    \end{remark}

\section{Experiments}\label{sec:experiments}
    In this section, the mapping-learning performance of the proposed approach are evaluated on realistic synthetic data. Additionally, its performance is compared against classical architecture from the INR literature.

    \subsection{Learning framework}
    The dataset generation can be divided into two phases: location generation and channel computation at those locations.

    \noindent\textbf{System parameters.} For all datasets, the central frequency (i.e. the reference frequency $f_r$) is set to $3.5$ GHz and the considered bandwidth is $50$ MHz. The BS is equipped with a uniform linear array (ULA) with half reference wavelength spacing. The number of considered antennas $N_a$ and frequencies $N_s$ is presented for each experiment.  

    \noindent\textbf{Location generation.} For all datasets, a $L$ by $L$ square scene, with $L=10$ m, is considered as the location space. Inside this scene, train locations are uniformly sampled with a certain location density. The uniform sampling influences the location generation process: when considering $N$ train locations, the train location matrix can be expressed as $\mathbf{X} = \left\{\mathbf{x}_i\right\}_{i=1}^{N} \in \mathbb{R}^{N \times 2}$, with:
    \begin{equation}
        \forall i \in \llbracket 1,N \rrbracket, \text{ } \mathbf{x}_i \sim \left(\mathcal{U}\left(-L/2, L/2\right), \mathcal{U}\left(-L/2,L/2\right)\right).
    \end{equation}
    The spatial density impacts the number of generated locations within the scene. The $L$ by $L$ square scene covers an area $S_m$ m$^2$, or equivalently $S_{\lambda}$ $\lambda^2$. To achieve a desired spatial density of $\gamma_m$ locs./m$^2$, or equivalently $\gamma_{\lambda}$ locs./$\lambda^2$, the following formula is applied:
    \begin{equation}
    N = S_m \gamma_m = S_{\lambda} \gamma_{\lambda}.
    \end{equation}
    Test locations are generated across a $2$D uniform grid within the scene. The spacing between each test location in both the $x$ and $y$ directions is set to $\lambda/4$. In the considered scene and at the considered reference frequency, this grid yields around $210$k test locations, enabling a thorough assessment of the mapping learning capabilities. Additionally, the spatial proximity of test locations allows to assess the learning performance of small scale channel fading.
    \begin{remark}
        Note that while all previous theoretical developments are proven in the $\mathbb{R}^3$ scenario, the following simulations of the location-to-channel mapping learning are obtained in $\mathbb{R}^2$ due to computational complexity. Hence, for all datasets, the locations are considered on the same elevation plane as the BS, i.e. considering $2$D locations and azimuth DoDs. Consequently, spatial frequencies are sampled from the $\mathcal{C}_1$ unit circle.    
    \end{remark}

    \noindent\textbf{Channel generation.} Realistic synthetic channels are generated using the Sionna ray-tracing library~\cite{sionna}. For each train/test location, ray-tracing techniques find the propagation paths and compute the channel coefficients. Three different datasets are considered:
    \begin{itemize}
        \item $\texttt{D}_1$: complex dataset with variable number of propagation paths regarding the considered location in the scene.
        \item $\texttt{D}_2$: $\texttt{D}_1$ with no LoS path.
        \item $\texttt{D}_3$: LoS dataset.
    \end{itemize}
    $\texttt{D}_2$ enables investigation of the mapping-learning capabilities of the proposed model under conditions characterized by an equitable distribution of path contributions, where no single path exhibits dominant power. As depicted in Fig.~\ref{fig:sionna_scene}, $\texttt{D}_1$ and $\texttt{D}_2$ are generated in the \textit{Etoile} Sionna scene in Paris, while $\texttt{D}_3$ is generated in an empty scene. For all datasets, diffraction and scattering are not considered. In $\texttt{D}_1$ and $\texttt{D}_2$, each propagation path can present at most two consecutive reflections.
    \begin{figure}
        \centering
        \includegraphics[scale=.25]{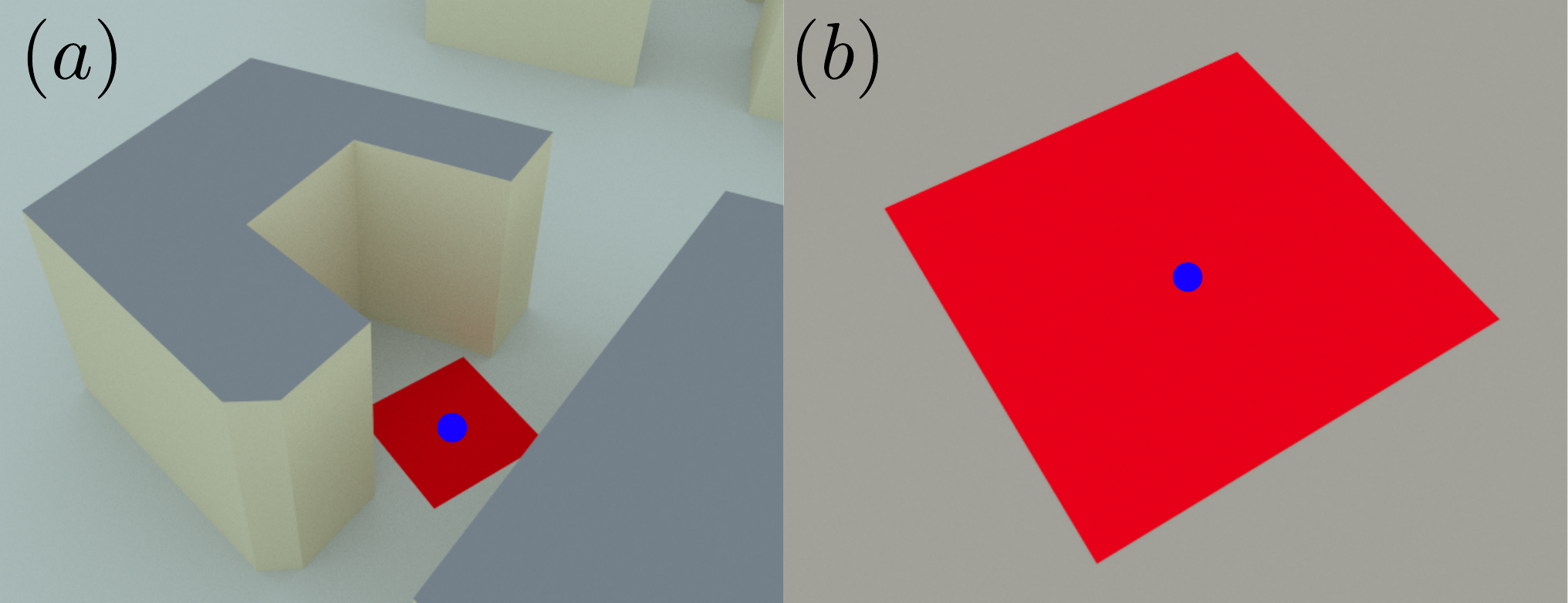}
        \caption{Ray-tracing scenes in Sionna: the red plane represents the possible train/test locations and the blue dot represents the BS. $(a)$: scene used for $\texttt{D}_1$ and $\texttt{D}_2$. $(b)$: scene used for $\texttt{D}_3$.}
        \label{fig:sionna_scene}
    \end{figure}


    \noindent\textbf{Training loss and evaluation metric.} All networks are trained using the classical $\ell_F$ loss, defined as:
    \begin{equation}\label{eq:loss}
        \mathcal{L} = \frac{1}{\abs{\mathcal{B}}}\sum_{\mathbf{x}\in\mathcal{B}} \norm{\mathbf{H}\left(\mathbf{x}\right)-f_{\boldsymbol{\theta}}\left(\mathbf{x}\right)}{F}^2,
    \end{equation}
    where $f_{\boldsymbol{\theta}}\left(\mathbf{x}\right) = \hat{\mathbf{H}}\left(\mathbf{x}\right)$ is the output of the neural network, and $\mathcal{B} \subset \mathbb{R}^2$ is the current batch set for the considered scene.
    The evaluation metric is the Normalized Mean Squared Error (NMSE) in dB over the test grid, defined as:
    \begin{equation}
        \hspace{-5pt}\text{NMSE}_{\left(\text{dB}\right)} = 10\log_{10}\left(\frac{1}{\abs{\mathcal{T}}}\sum_{\mathbf{x}\in\mathcal{T}}\dfrac{\norm{\mathbf{H}\left(\mathbf{x}\right)-f_{\boldsymbol{\theta}}\left(\mathbf{x}\right)}{F}^2}{\norm{\mathbf{H}\left(\mathbf{x}\right)}{F}^2}\right),
    \end{equation}
    where $\mathcal{T} \subset \mathbb{R}^2$ is the test location set for the considered scene.

    \begin{remark}
        While the proposed classical $\ell_F$ loss function may appear overly simplistic for this learning task, the subsequent experiments demonstrate that it achieves satisfactory learning results. A more complex loss function, which would incorporate the selection of the correct planar wavefronts/SVs/FRVs for each training location, could be envisioned. However, in contrast to the proposed $\ell_F$ loss, such approach would require a more extensive data collection process, as both the azimuth and elevation angle of arrivals would need to be gathered for each training location.
    \end{remark}

    \begin{remark}
        Given that the proposed model necessitates a training phase utilizing channel coefficients, obtained through ray-tracing in this study, a natural point of inquiry emerges regarding the rationale for not directly employing ray-tracing for channel prediction, given its applicability as an applied AI model. The primary theoretical and practical advantage of the proposed method over ray-tracing is that accurate channel prediction through ray-tracing necessitates the constitution of a digital twin, with rigorous modelization of the scene at the electromagnetic level. Specifically, this ray-tracing approach mandates the determination of each scene element's electromagnetic properties such as material permittivity or conductivity and reflection/diffractions coefficients. In contrast, the proposed method exclusively requires the collection of channel coefficients at different locations within a scene, and demonstrates independence from a priori knowledge of the scene topology. Furthermore, the ray-tracing approach exhibits a significantly superior computational complexity: predicting the channel at a given location involves shooting rays in all directions across the $\mathbb{R}^3$ space. Conversely, once trained, the proposed method exclusively involves computationally efficient matrix operations, as described in Eq.~\eqref{eq:global_matrix_form}.
    \end{remark}

    \subsection{Baselines}
    It is proposed to assess the performance of the proposed architecture against baseline models from the INR literature. As such, several baselines are proposed in Fig.\ref{fig:baselines}: 1. is a classical complex MLP while 2. and 3. are complex RFFs. Additionally, in 2., the spatial frequencies are drawn from a gaussian distribution, while 3. considers the same spatial frequencies as in the model-based architecture, i.e. sampled from the $\mathcal{C}_1$ circle. All MLPs output a vector of dimension $N_aN_s$ which is then reshaped to obtain the estimated channel matrix $\hat{\mathbf{H}}\left(\mathbf{x}\right) \in \mathbb{C}^{N_a \times N_s}$.
    \begin{figure}[h]
        \centering
        \includegraphics[scale=.55]{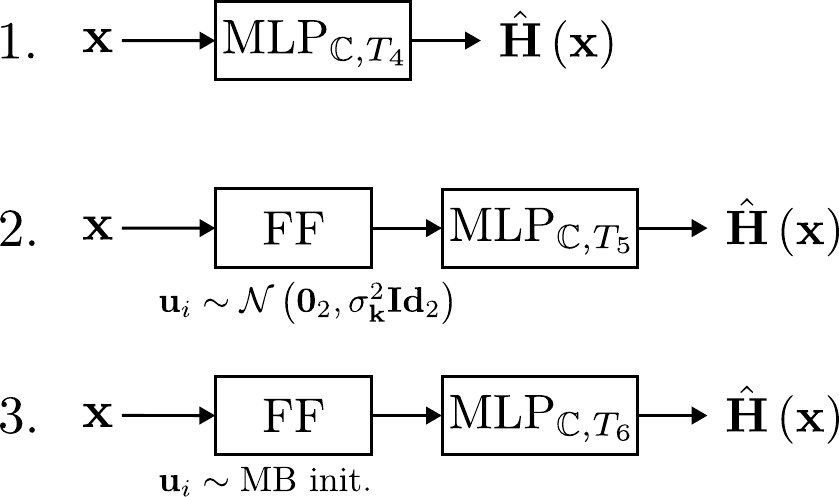}
        \caption{Baselines: 1. MLP, 2. and 3. RFFs.}
        \label{fig:baselines}
    \end{figure}

    \subsection{Experimental results}\label{subsec:experiment_results}
    For all experiments, the proposed networks and baselines have the following layer sizes: $T_1 = 256$, $T_2 = T_3 = 64$, $T_4 = 1024$, and $T_5 = T_6= 64$. Additionally, unless stated otherwise, $D=1000$ spatial frequencies are considered. Those values have been empirically chosen to maximize performance and minimize overfitting. For RFFs, increasing the number of learning parameters did not improve performance. Finally, a training location density of $175$ locs.$/$m$^2$ $\simeq 1.3$ locs.$/\lambda^2$, which is equivalent to $17.5$k training locations, is considered at the exception of the experiment depicted in Fig.~\ref{fig:NMSE_density_evolution}.

    \noindent\textbf{Scene reconstruction.} Let $N_a = 64, N_s = 64$. It is proposed to study the ability of each network to learn the location-to-channel mapping in different radio-environments. 
    
    Table~\ref{table:nmse_inference} presents the reconstruction results. One can remark that the proposed MB-$\tilde{\mathbf{\Psi}}_{\mathbf{a}}$ model outperforms every baseline and the other MB model in every propagation scene. The performance gap observed between the MB-$\tilde{\mathbf{\Psi}}_{\mathbf{a}}$ and MB-$\tilde{\mathbf{u}}$ models is due to an hypothesis made in the proof of Proposition~\ref{proposition:channel_approx} and is explained in a subsequent experiment. The bad mapping learning capabilities of the baselines can be explained by the lack of structure in their architecture.

    Fig.~\ref{fig:real_reconstruct} presents the reconstructed channel after training, for several models. The top row represents the real part of the reconstructed channels, while the bottom row represents the associated spatial spectrums (i.e. $2$D Fourier transform). Note that the vertical and horizontal lines in the bottom row are artifacts arising from the rectangular windowing of the scene. It clearly appears that the MB network efficiently learns the desired mapping. Additionally, one can remark that the RFF model presents high-frequency spatial contents due to its embedding stage, as shown on the bottom row, but fails to learn the complex structure of the propagation channel. Finally, the reconstructed channel for the MLP model presents low-frequency spatial contents: the spectral content is concentrated near the origin on the bottom row. This behavior illustrates the spectral-bias issue inherent to this architecture.

    Fig.~\ref{fig:NMSE_density_evolution} presents the mapping-learning performance evolution with respect to the training location density on the $\texttt{D}_1$ dataset. One can remark that the MB-$\tilde{\mathbf{\Psi}}_{\mathbf{a}}$ network outperforms the baseline in any density configuration. Additionally, the proposed network presents a failure mode in the low location density regime. In this regime, the scarcity of spatially close training locations results in the learning failure of the rapidly varying spatial content. Note that the proposed method achieves quasi perfect reconstruction in sub-Shannon-Nyquist location density, as the $2$D Shannon-Nyquist criterion yields a density of $4$ locs.$/\lambda^2$.

    Fig.~\ref{fig:NMSE_ant_freq} presents the ground-truth and estimated frequency/antenna responses for a given test location in the $\texttt{D}_1$ dataset. One can see that, due to the multipath nature of the $\texttt{D}_1$ scene, the frequency response presents rapid variations. Both the estimated frequency and antenna responses are close to the ground-truth highlighting the very good reconstruction performance of the proposed model network.

    \begin{table}[t]
        \caption{$\text{NMSE}$ (in dB) over the test grid.}
        \centering
            \begin{tabular}{lcccccc}
                \toprule
                &  MLP & RFF & RFF (MB init.) & \textbf{MB}-$\tilde{\mathbf{\Psi}}_{\mathbf{a}}$ & MB-$\tilde{\mathbf{u}}$\\
                \midrule
                Params. & $10.5$M & $669$k & $669$k & $\mathbf{9.1}$\textbf{M} & $851$k \\
                \toprule
                $\texttt{D}_1$ & $0.01$ & $0.02$ & $2.10^{-3}$ & $\mathbf{-29.23}$ & $-14.60$\\
                \midrule
                $\texttt{D}_2$ & $0.01$ & $0.05$ & $3.10^{-3}$ & $\mathbf{-20.19}$ & $-10.25$\\
                \midrule
                $\texttt{D}_3$ & $0.01$ & $0.02$ & $2.10^{-3}$ & $\mathbf{-40.67}$ & $-11.40$\\
                \bottomrule
            \end{tabular}
        \label{table:nmse_inference}
    \end{table}

    \begin{figure*}
        \centering
        \includegraphics[width=2\columnwidth]{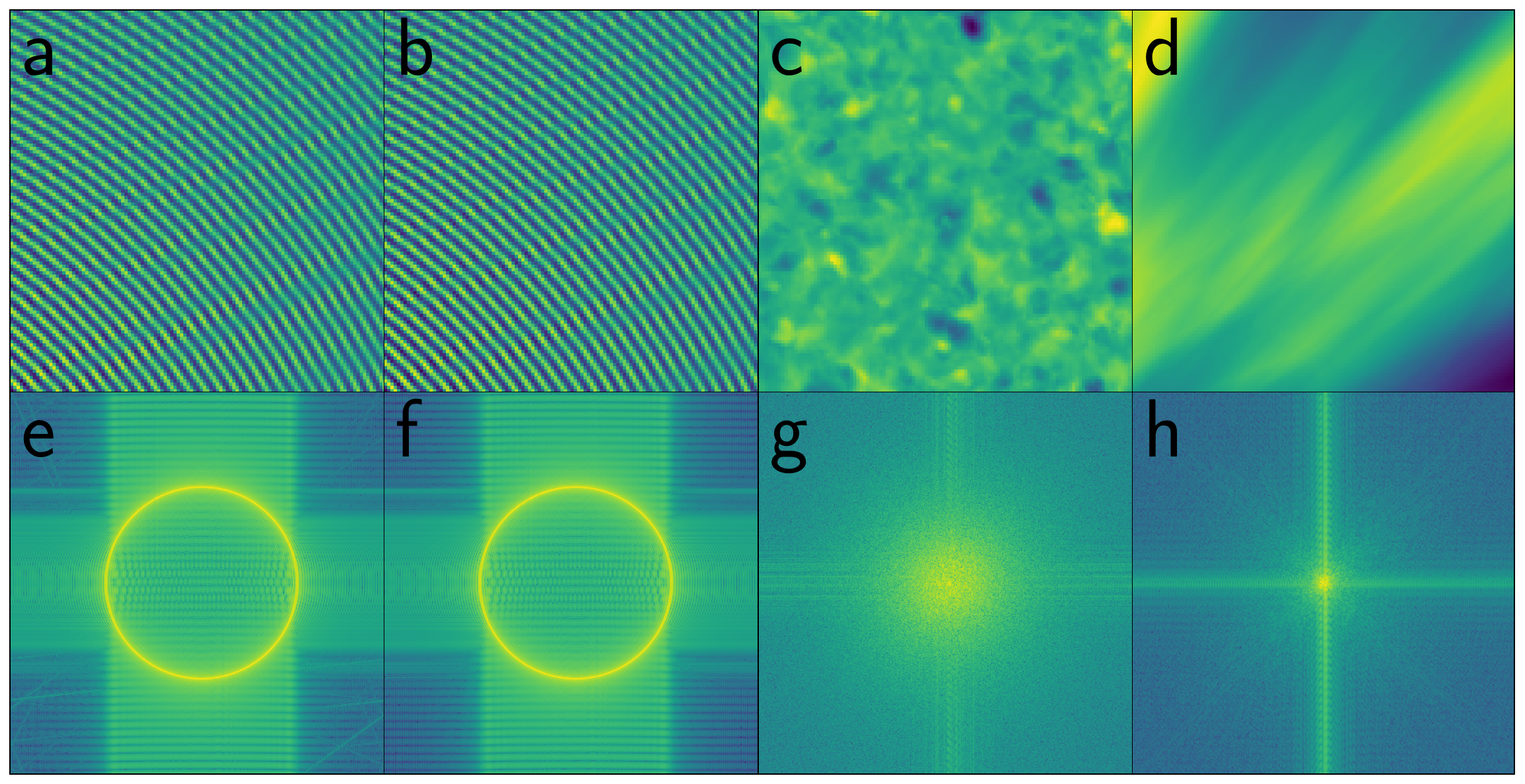}
        \caption{Top row: real part of reconstructed channel over a $2.5$m by $2.5$m zone of $\texttt{D}_1$, for the central antenna and frequency ($N_a =64$, $N_s=64$). Bottom row: spatial spectrum of the reconstructed channel over the entire $\texttt{D}_1$ scene (logarithmic scaling), for the central antenna and frequency ($N_a =64$, $N_s=64$). $\textsf{a}$/$\textsf{e}$: Ground Truth, $\textsf{b}$/$\textsf{f}$: MB-$\tilde{\mathbf{\Psi}}_{\mathbf{a}}$, $\textsf{c}$/$\textsf{g}$: RFF, $\textsf{d}$/$\textsf{h}$: MLP.}
        \label{fig:real_reconstruct}
    \end{figure*}

    \begin{figure*}[htbp]
        \centering
        \begin{minipage}[b]{0.49\textwidth}
            \centering
            \includegraphics[scale=.6]{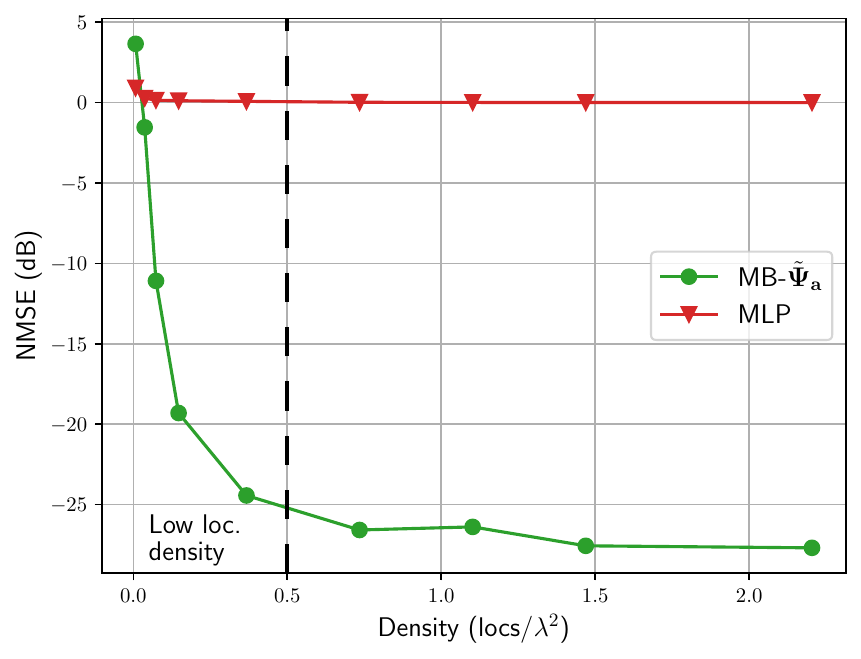}
            \caption{NMSE evolution with the training location density: MB-$\tilde{\mathbf{\Psi}}_{\mathbf{a}}$ learning, $\texttt{D}_1$ dataset ($N_a=64$, $N_s=64$).}
            \label{fig:NMSE_density_evolution}
            \vspace{0.3cm}
            \includegraphics[scale=.35]{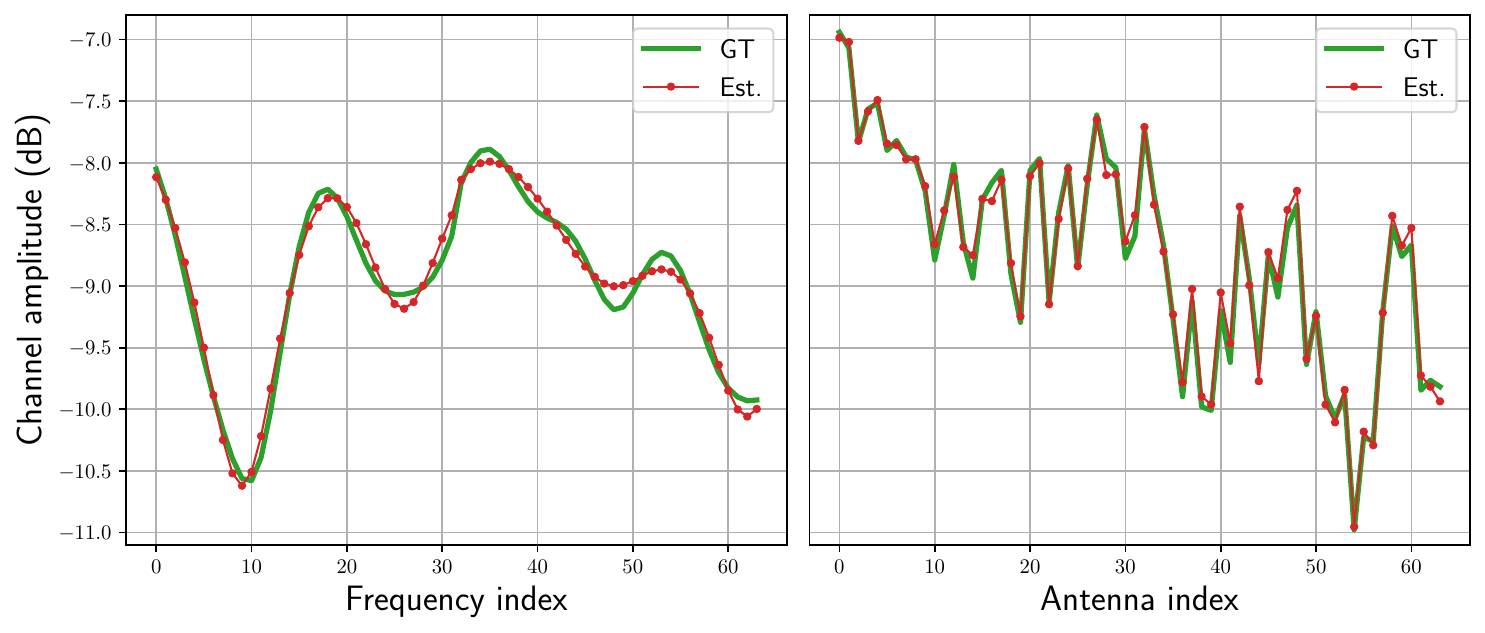}
            \caption{Reconstruction performance for a given location $\mathbf{x}_0$:  MB-$\tilde{\mathbf{\Psi}}_{\mathbf{a}}$ learning, $\texttt{D}_1$ dataset ($N_a=64$, $N_s=64$).}
            \label{fig:NMSE_ant_freq}
        \end{minipage}
        \hfill
        \begin{minipage}[b]{0.49\textwidth}
            \centering
            \includegraphics[scale=.35]{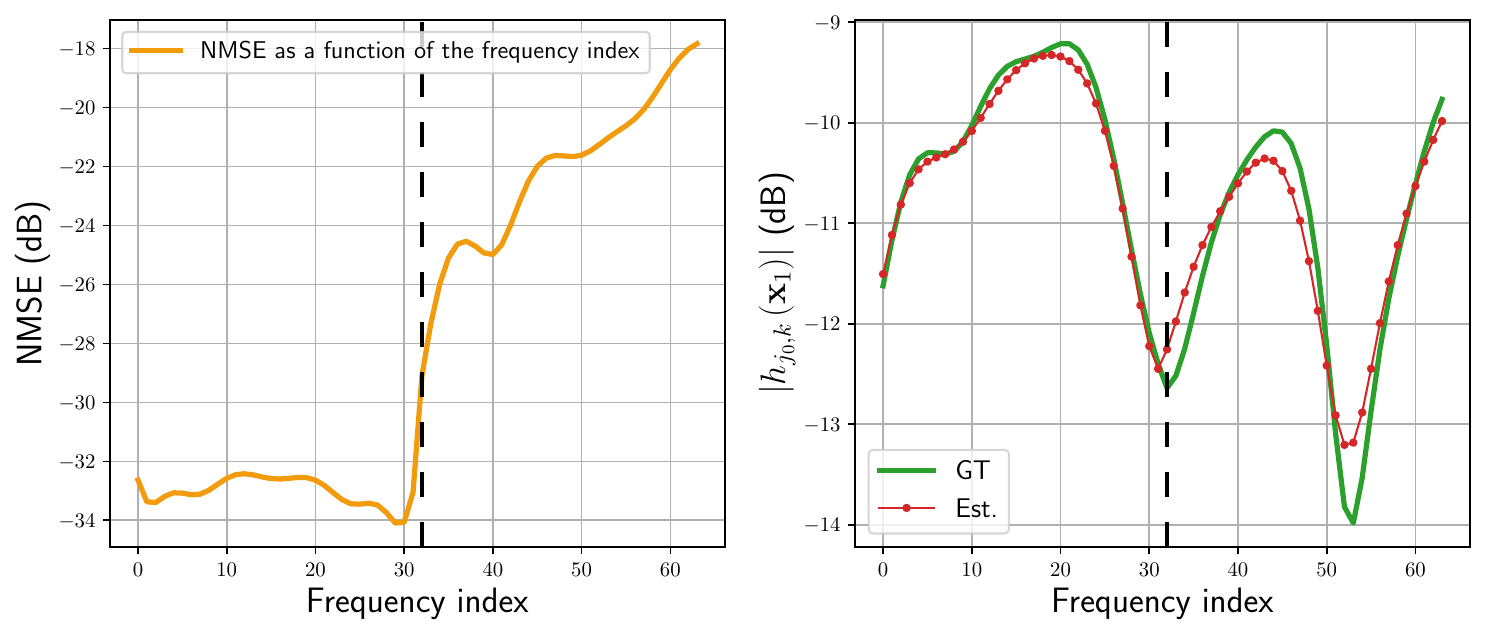}
            \caption{Frequency generalization performance: MB-$\tilde{\mathbf{\Psi}}_{\mathbf{a}}$ learning, $\texttt{D}_1$ dataset ($N_a=64$, $N_s=32$ for training, $N_s=64$ for testing).}
            \label{fig:freq_extrapol}
            \vspace{0.3cm}
            \includegraphics[scale=.6]{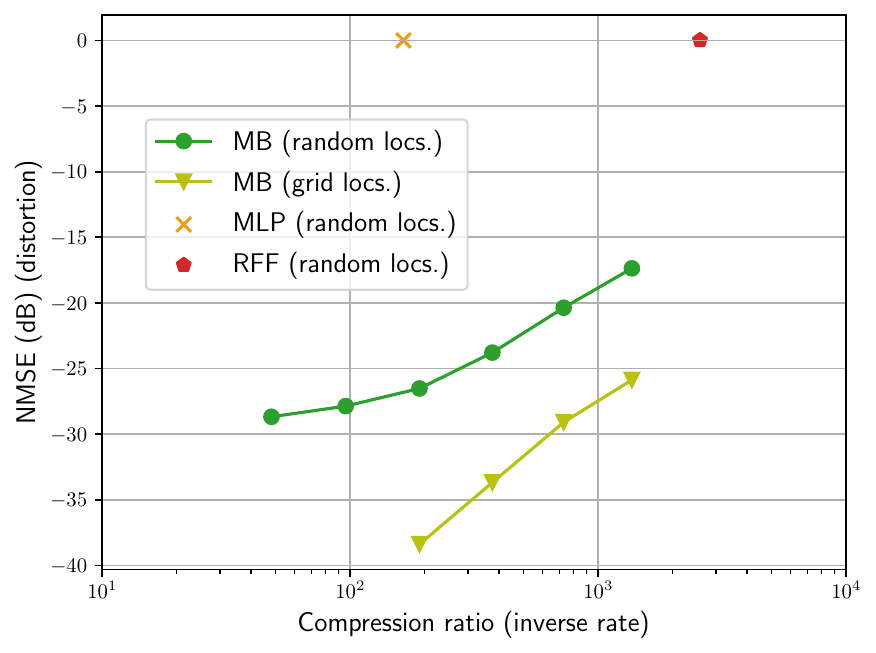}
            \caption{Learning performance evolution with the compression ratio: MB-$\tilde{\mathbf{\Psi}}_{\mathbf{a}}$ learning, $\texttt{D}_1$ dataset ($N_a=64$, $N_s=64$).}
            \label{fig:rate_distortion}
        \end{minipage}
    \end{figure*}
    
    \noindent\textbf{Generalization capabilities.} Fig.~\ref{fig:freq_extrapol} presents the frequency generalization performance of the MB-$\tilde{\mathbf{\Psi}}_{\mathbf{a}}$ model trained on the $\texttt{D}_1$ dataset with $N_s = 32$, and evaluated on $N_s=64$ frequencies. The dashed line in Fig.~\ref{fig:freq_extrapol} represents the train/test frequency separation. The NMSE, computed over the test locations, for each frequency at the central antenna, is unsurprisingly higher for the frequency unseen during frequency but still remains low. It highlights the generalization capabilities of the model-based network: as the FRV hypernetwork learns propagation delays, it successfully learns the physics of the propagation scene and is agnostic to the considered frequencies. This is illustrated on the right side of Fig.~\ref{fig:freq_extrapol}, where the reconstructed channel for a given location $\mathbf{x}_1$ is almost perfect on the frequencies used during training, and still follows the ground-truth for frequencies not seen during training.

    \noindent\textbf{Radio-environment compression.} Let $N_a = 64, N_s = 64$. Let us consider the transmission (or storage) of the test dataset: in addition to the location information, one has to transmit $2 N_a N_s N_l$ real numbers where $N_l$ is the number of locations. As the test location grid is very dense ($N_l \simeq 210$k), it could be more efficient to only transmit the weights of the trained model-based architecture. Indeed, as the MB-$\tilde{\mathbf{\Psi}}_{\mathbf{a}}$ model achieves good location-to-channel mapping learning, one can use this network and the transmitted locations to reconstruct the channel. As an example, when considering encoding of variables on $32$ bits, the channel coefficients over the test grid weight around $6.9$ Go while the $9.1$M learning parameters of the MB-$\tilde{\mathbf{\Psi}}_{\mathbf{a}}$ model (with $D=1000$) only weight around $36.4$ Mo. Let $N_b$ be the number of real learnable coefficients of the proposed neural architecture, the compression ratio is then defined as $R=2 N_a N_s N_l/N_b$.

    Fig.~\ref{fig:rate_distortion} presents the evolution of the learning performance with respect to the compression ratio $R$ for various networks. It can be shown that the number of sampled spatial frequencies $D$ is predominant on the learning error for the MB models. Therefore, it is proposed to only vary this parameter to vary $N_b$. The random locs. approach refers to the training on random locations with $1.3$ locs.$/\lambda^2$ spatial density, while the grid locs. refers to the training on the entire $\lambda/4$ uniform grid. One can remark in Fig.~\ref{fig:rate_distortion} that, even at high compression ratios, the NMSE of the MB-$\tilde{\mathbf{\Psi}}_{\mathbf{a}}$ model remains significantly low, highlighting the efficiency of the proposed approach for radio-environment compression. One can also note that, when considering the grid locs. approach, the model-based network reaches very good reconstruction performance, with a NMSE performance increase by around $10$dB in comparison to the random locs. approach. Therefore, it enables the transmission of the dataset through the network weights with minimal reconstruction error, at the expense of a longer training phase. One should note that the compression ratio is virtually infinite. Indeed, as the trained network efficiently learns the location-to-channel mapping, it is able to infer the channel matrix for any location in the considered scene with good performance. Finally, one can interpret this experiment under the rate-distortion paradigm: in this context, the distortion measure is represented by the NMSE in dB, while the rate is represented as the inverse of the compression ratio. As the compression ratio increases (i.e. the rate decreases), the proposed neural model is required to encode the channels using fewer learnable parameters, resulting in a higher NMSE, which indicates an increase in distortion. This behavior aligns with the rate-distortion theory, where the trade-off between distortion and rate is investigated: in Fig.~\ref{fig:rate_distortion}, it can be seen seen that the distortion is nearly logarithmically proportional to the rate inverse. This implies that a slight decrease in rate (or equivalently, an increase in the compression ratio) results only in a modest increase in distortion.

    \noindent\textbf{Performance evolution with $N_a$ and $N_s$.} Fig.~\ref{fig:freq_side_effects} exposes the performance evolution along the antenna array and system frequencies for the MB-$\tilde{\mathbf{u}}$ learning network, trained on the $\texttt{D}_1$ dataset with $N_a = 64$, $N_s=64$. It presents the NMSE along the antenna indexes, computed over the test grid at the central frequency, and the NMSE along the frequency indexes, computed over the test grid at the central antenna. One can observe that, while the NMSE is almost constant over the different frequencies, it rises on the antenna array sides. This can be explained through the previous theoretical developments: the model-based network has been designed with two Taylor expansions, one on the locations and another on the antennas. When the considered antenna array is large, the antenna correction terms fail for the array sides due the local nature of the Taylor expansion. This results in an increased error response on the antenna array sides. On the other hand, as the frequency correction terms do not originate from an approximation (see Eq.~\eqref{eq:app_intermedary} in Appendix~\ref{appendix:c}), the error response over the different system frequencies does not exhibit error spikes. It motivates the use of the MB-$\tilde{\mathbf{\Psi}}_{\mathbf{a}}$ model which drops the model-based constraints for the FRV dictionary, overcoming this side effect. Note that a more extensive comparison between the two architectures is presented in the following experiment.

    \begin{figure*}[htbp]
        \centering
        \begin{minipage}[b]{0.49\textwidth}
            \centering
            \includegraphics[scale=.5]{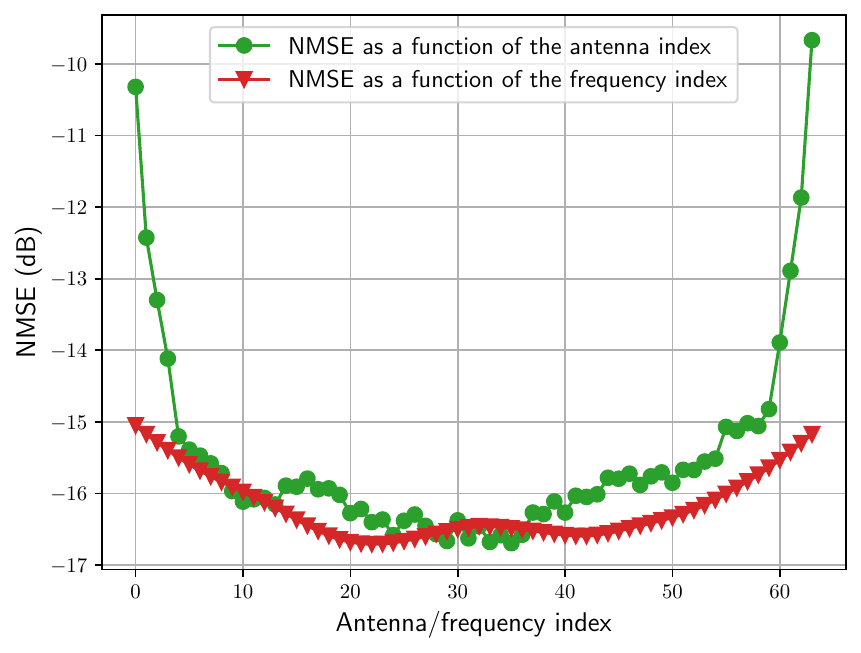}
            \caption{NMSE evolution with antenna/frequency index: MB-$\tilde{\mathbf{u}}$ learning,\\ $\texttt{D}_1$ dataset ($N_a=64$, $N_s=64$).}
            \label{fig:freq_side_effects}
            \vspace{0.3cm}
            \includegraphics[scale=.5]{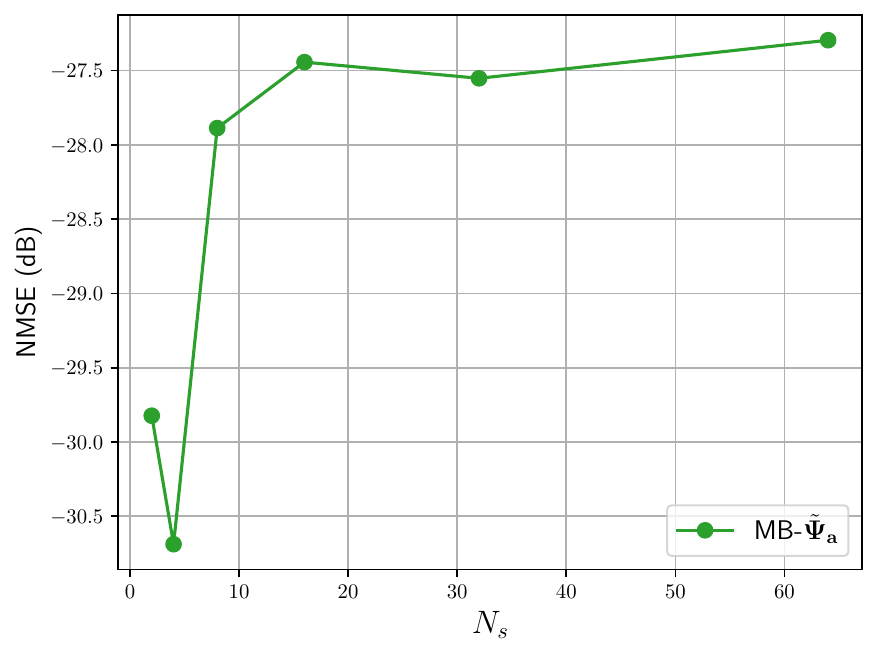}
            \caption{NMSE evolution with $N_s$: MB-$\tilde{\mathbf{\Psi}}_{\mathbf{a}}$ learning, $\texttt{D}_1$ dataset ($N_a=2$).}
            \label{fig:Ns_NMSE}
        \end{minipage}
        \hfill
        \begin{minipage}[b]{0.49\textwidth}
            \centering
            \includegraphics[scale=.7]{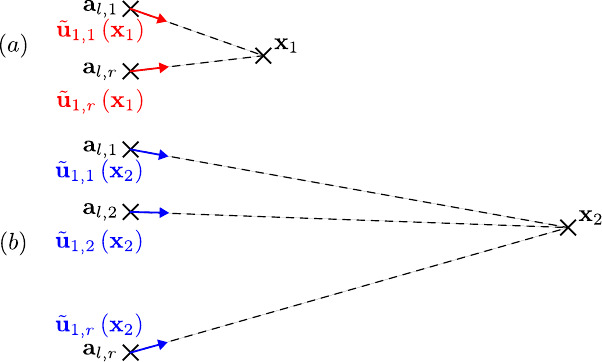}
            \caption{Failure scenarios for the DoD equality hypothesis: $\left(a\right)$ locations close to the antenna array, $\left(b\right)$: locations far from a large antenna array.}
            \label{fig:DoD_hypo}
            \vspace{0.3cm}
            \includegraphics[scale=.5]{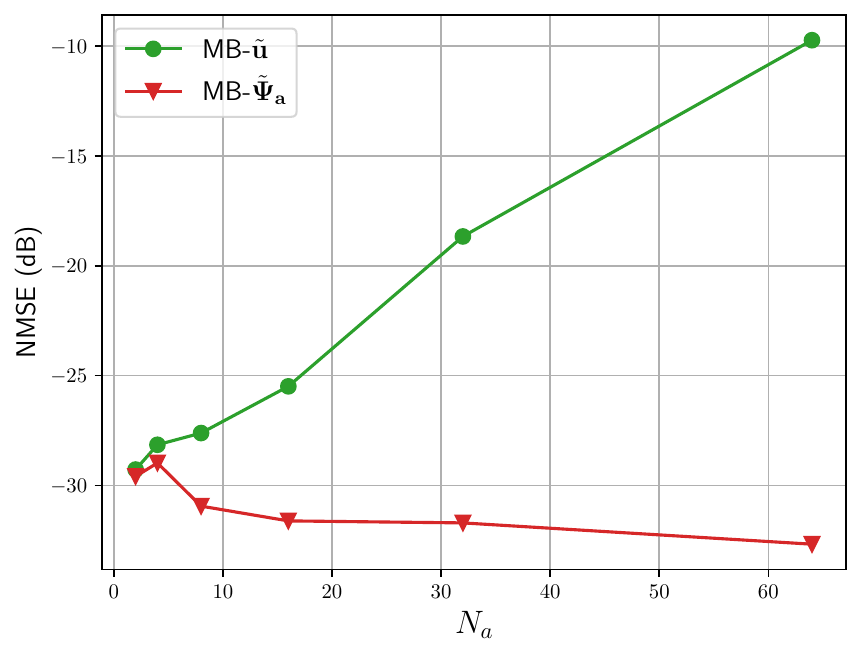}
            \caption{NMSE evolution with $N_a$: $\texttt{D}_3$ dataset ($N_s=2$).}
            \label{fig:T4_mb_utilde}
        \end{minipage}
    \end{figure*}

    Fig.~\ref{fig:Ns_NMSE} presents the performance evolution with respect to the number of frequencies by setting $N_a = 2$ and letting $N_s$ vary over the fixed $50$ MHz bandwidth, for the MB-$\tilde{\mathbf{\Psi}}_{\mathbf{a}}$ model. One can remark that the NMSE curve presents a minor performance degradation when $N_s$ increases from $1$ to $8$ and then remains constant for higher values of $N_s$. As the FRV hypernetwork is independent on the number of frequencies, increasing the number of frequencies is equivalent to considering a harder learning problem with the same number of learning parameters. This explains the minor performance degradation when $N_s$ increases. However note that, even with $N_s = 64$, the NMSE performance remains satisfactory. Additionally, the adaptability of the proposed architecture is enhanced: in~\cite{chatelier2023modelbased}, the proposed network was only able to learn the location-to-channel mapping for a unique antenna and frequency. Directly transposing the model-based architecture of~\cite{chatelier2023modelbased} for multi-antenna/multi-frequency would require one parallel model for each wanted antenna/frequency pair. In this paper, the proposed model-based architecture directly scales with the number of frequencies: as the FRV hypernetwork is independent on $N_s$, one can use the same network for any wanted $N_s$ with the same learning-parameter complexity. The only complexity difference comes from Kronecker product of Eq.~\eqref{eq:global_matrix_form}, which is negligible compared to the total numbers of operations in the neural network.

    \noindent\textbf{MB-$\tilde{\mathbf{\Psi}}_{\mathbf{a}}$ learning vs MB-$\tilde{\mathbf{u}}$ learning.} During the theoretical developments of Proposition~\ref{proposition:channel_approx}, the hypothesis of DoD equality over the entire antenna array has been made to obtain Eq.~\eqref{eq:chan_approx}. This hypothesis fails in two scenarios, presented in Fig.~\ref{fig:DoD_hypo}: when the considered scene is close to the antenna array, and when the considered scene is far from a large antenna array.

    Fig.~\ref{fig:T4_mb_utilde} presents the performance evolution with respect to the number of antennas for $N_s = 2$ and both the MB-$\tilde{\mathbf{\Psi}}_{\mathbf{a}}$ and MB-$\tilde{\mathbf{u}}$ models trained on the $\texttt{D}_3$ dataset. One can remark that the NMSE increases with the number of antennas for the MB-$\tilde{\mathbf{u}}$ model: this is an immediate consequence of the aforementioned failure scenarios. As the propagation scene for the $\texttt{D}_3$ dataset consists of a small zone around the BS, when one increases the number of antennas, the DoD equality across antennas does not hold anymore, resulting in a performance loss. This phenomenon is absent for the MB-$\tilde{\mathbf{\Psi}}_{\mathbf{a}}$ model: in this architecture, the entire FRV matrix is learned, overcoming the DoD assumption issue.

    \subsection{Areas of applications}
    Previous experiments assumed perfect knowledge of the location where the channel needed to be evaluated. However, if an offset arises from an imperfect location estimation process, the predicted channel from the trained model will correspond to the channel at the offset location. Since the channel varies rapidly with the location (on the order of the wavelength), the error between the true channel and the predicted channel will be significant, thereby limiting the application of the proposed method for precise channel estimation. Notwithstanding, the proposed method can be considered for different applications, some of which are outlined below.
    
    \noindent\textbf{Radio-environment compression.} As discussed in Section~\ref{subsec:experiment_results}, the model's ability to learn the location-to-channel mapping enables the storage of channel coefficients within its learnable parameters. Fig.~\ref{fig:rate_distortion} illustrates the reconstruction performance as a function of the compression ratio, demonstrating that the test dataset can be compressed up to a compression ratio of $10^3$ without significant reconstruction error.

    \noindent\textbf{UE localization.} Let $\mathbf{x}_1 \in \mathbb{R}^3$ be a UE physical location and $\mathbf{H}\left(\mathbf{x}_1\right) \in \mathbb{C}^{N_a \times N_s}$ the corresponding channel coefficients. One might inquire the following question: \textit{Is it possible to estimate $\mathbf{x}_1$ using the trained neural architecture?} As the proposed trained neural model achieves the location-to-channel mapping learning at the wavelength level, it has the potential to enable a highly precise localization process. Two different localization approaches that uses the trained model are outlined below.
    \begin{itemize}
        \item \textit{Gradient-descent approach.} Although directly inverting the forward pass of the trained model to perform localization is not feasible, the differentiability of the model facilitates an alternative localization process summarized below:
        \begin{enumerate}
            \item Fix the neural model learnable parameters.
            \item Define a new learnable parameter $\tilde{\mathbf{x}} \in \mathbb{R}^{3}$.
            \item Estimate $\mathbf{H}\left(\tilde{\mathbf{x}}\right)$ with the forward pass of the trained neural model: $f_{\boldsymbol{\theta}}\left(\tilde{\mathbf{x}}\right)$.
            \item Update $\tilde{\mathbf{x}}$ through several gradient descent steps: $\tilde{\mathbf{x}} \leftarrow \tilde{\mathbf{x}} - \gamma \nabla_{\tilde{\mathbf{x}}} \norm{\mathbf{H}\left(\mathbf{x}_1\right) - f_{\boldsymbol{\theta}}\left(\tilde{\mathbf{x}}\right)}{F}$. The goal is to minimize $\norm{\mathbf{H}\left(\mathbf{x}_1\right) - f_{\boldsymbol{\theta}}\left(\tilde{\mathbf{x}}\right)}{F}$ so that $\tilde{\mathbf{x}} \rightarrow \mathbf{x}_1$.
        \end{enumerate}
        \item \textit{Local exhaustive search approach.} It is assumed that, in addition to the channel coefficients $\mathbf{H}\left(\mathbf{x}_1\right)$, a rough estimate of $\mathbf{x}_1$, namely $\hat{\mathbf{x}}_1$ is known. It is then possible to perform localization using a location grid and the trained model, as summarized below:
        \begin{enumerate}
            \item Generate a location grid $\left\{\tilde{\mathbf{x}}_i\right\}_{i=1}^N$ around $\hat{\mathbf{x}}_1$.
            \item Estimate the channel at each grid location using the trained model: $\left\{f_{\boldsymbol{\theta}}\left(\tilde{\mathbf{x}}_i\right)\right\}_{i=1}^N$.
            \item Update the estimated location: $\hat{\mathbf{x}}_1 \leftarrow \tilde{\mathbf{x}}_{i^\star}$ where $i^\star = \argmin_{i} \norm{\mathbf{H}\left(\mathbf{x}_1\right) - f_{\boldsymbol{\theta}}\left(\tilde{\mathbf{x}}_i\right)}{F}$.
        \end{enumerate}
    \end{itemize}

    \noindent\textbf{General communication tasks.} Upon training completion, the BS could use this neural model to obtain a rough estimate of the channel coefficients for different UEs, without the need of sending pilot symbols. This could be useful for many tasks: efficient beam management, resource allocation, interference management, secure communication mechanisms\dots

\section{Conclusion and future work}\label{sec:conclusion}
    This paper presented a study on the location-to-channel mapping learning problem. Through analytical developments based on Taylor expansions of the propagation distance, a model-based neural architecture was proposed. Its performance have been studied on realistic channels against several baselines from the Implicit Neural Representation literature, showing great mapping learning performance in several realistic propagation scenes. It also showed that the proposed architecture overcame the spectral-bias issue. Additionally, the use of the proposed network for radio-environment compression has been exposed. Transmitting the parameters of the trained model-based model and reconstructing the channel coefficients at different locations has been proven to be more efficient than directly transmitting the channel coefficients, with compression ratio reaching $10^3$ without major performance loss. Finally, the theoretical developments made to obtain the approximated channel model have been leveraged to explain the model-based architecture performance in several scenarios, demonstrating the great interpretability of the model-based machine learning paradigm.

    Future work will consider the optimization of the proposed architecture learning-parameter number. This could be done by reducing the number of planar wavefronts and optimizing them through an update rule or gradient descent. Future work could also consider the proposed method in a time-varying scene, i.e. with mobility. It may be possible to train the scene from channel samples in the static scenario and then use an online learning strategy to continuously adapt the learned mapping with channel samples taking into account mobility.

\bibliographystyle{IEEEtran}
\bibliography{./refs/biblio.bib}

\appendices

\section{Proof of Lemma~\ref{lemma:taylor_approx}}\label{apendix:taylor}
    Let $\mathbf{x}_r \in \mathbb{R}^3$ be a reference location and $\mathcal{D}_{\mathbf{x}} \subset \mathbb{R}^3$ be a local validity domain such that $\forall \mathbf{x} \in \mathcal{D}_{\mathbf{x}}, \norm{\mathbf{x}-\mathbf{x}_r}{2} \leq \epsilon_{\mathbf{x}}$. The same holds true in the antenna subspace with the reference antenna location $\mathbf{a}_{l,r} \in \mathbb{R}^3$ and local validity domain $\mathcal{D}_{\mathbf{a}} \subset \mathbb{R}^3$ such that $\forall \mathbf{a}_{l,j} \in \mathcal{D}_{\mathbf{a}}, \norm{\mathbf{a}_{l,j}-\mathbf{a}_{l,r}}{2} \leq \epsilon_{\mathbf{a}}$. Let $\xi \left(\mathbf{x},\mathbf{a}_{l,j}\right) \triangleq \norm{\mathbf{x}-\mathbf{a}_{l,j}}{2}$. $\xi \left(\mathbf{x},\mathbf{a}_{l,j}\right)$ is differentiable at $\mathbf{x} = \mathbf{x}_r$. The first order Taylor expansion of $\xi \left(\mathbf{x},\mathbf{a}_{l,j}\right)$ around $\mathbf{x}_r$ yields, $\forall \mathbf{x} \in \mathcal{D}_{\mathbf{x}}$:
    \begin{align}
        \xi\left(\mathbf{x},\mathbf{a}_{l,j}\right) &\simeq \xi\left(\mathbf{x}_r,\mathbf{a}_{l,j}\right) + \left. \nabla_{\mathbf{x}} \norm{\mathbf{x}-\mathbf{a}_{l,j}}{2} \right|_{\mathbf{x}_r}^\transp \left(\mathbf{x}-\mathbf{x}_r\right)\nonumber\\
        &= \norm{\mathbf{x}_r -\mathbf{a}_{l,j}}{2} + \dfrac{\left(\mathbf{x}_r - \mathbf{a}_{l,j}\right)^\transp}{\norm{\mathbf{x}_r -\mathbf{a}_{l,j}}{2}}\left(\mathbf{x}-\mathbf{x}_r\right)\nonumber\\
        &= \norm{\mathbf{x}_r -\mathbf{a}_{l,j}}{2} + \mathbf{u}_{l,j}\left(\mathbf{x}_r\right)^\transp\left(\mathbf{x}-\mathbf{x}_r\right).\label{eq:intermerdary_ref_dist}
    \end{align}
    One has $\xi\left(\mathbf{x}_r,\mathbf{a}_{l,j}\right) = \norm{\mathbf{x}_r -\mathbf{a}_{l,j}}{2}$. $\xi \left(\mathbf{x}_r,\mathbf{a}_{l,j}\right)$ is differentiable at $\mathbf{a}_{l,j} = \mathbf{a}_{l,r}$. As per the previous development, the first order Taylor expansion of $\xi \left(\mathbf{x}_r,\mathbf{a}_{l,j}\right)$ around $\mathbf{a}_{l,r}$ yields, $\forall \mathbf{a}_{l,j} \in \mathcal{D}_{\mathbf{a}}$:
    \begin{equation}
        \xi \left(\mathbf{x}_r,\mathbf{a}_{l,j}\right) \simeq \norm{\mathbf{x}_r-\mathbf{a}_{l,r}}{2} - \mathbf{u}_{l,r}\left(\mathbf{x}_r\right)^\transp \left(\mathbf{a}_{l,j}-\mathbf{a}_{l,r}\right).
    \end{equation}
    Finally one has: $\forall \left(\mathbf{x},\mathbf{a}_{l,j}\right) \in \mathcal{D}_{\mathbf{x}} \times \mathcal{D}_{\mathbf{a}}: \norm{\mathbf{x}-\mathbf{a}_{l,j}}{2} \simeq \norm{\mathbf{x}_r-\mathbf{a}_{l,r}}{2} + \mathbf{u}_{l,j}\left(\mathbf{x}_r\right)^\transp\left(\mathbf{x}-\mathbf{x}_r\right) - \mathbf{u}_{l,r}\left(\mathbf{x}_r\right)^\transp \left(\mathbf{a}_{l,j}-\mathbf{a}_{l,r}\right)$.

\section{Proof of Corollary~\ref{corollary:error_approx}}\label{appendix:cor_1_proof}
    Let $\mathbf{x}_r \in \mathbb{R}^3$ and $\mathbf{a}_{l,r} \in \mathbb{R}^3$ be a reference location and a reference antenna location respectively. Let $\mathbf{\Gamma}_{\mathbf{x}_r} \triangleq \left(\mathbf{x}_r - \mathbf{a}_{l,j}\right)\left(\mathbf{x}_r - \mathbf{a}_{l,j}\right)^\transp$, and $\mathbf{\Gamma}_{\mathbf{a}_r} \triangleq \left(\mathbf{a}_{l,r} - \mathbf{x}_r\right)\left(\mathbf{a}_{l,r} - \mathbf{x}_r\right)^\transp$. Let $\psi\left(\mathbf{x}\right)$ be a double differentiable function at $\mathbf{x} = \mathbf{x}_r$. The second order of the Taylor expansion of $\psi\left(\mathbf{x}\right)$ can be expressed as $e = \frac{1}{2} \left(\mathbf{x}-\mathbf{x}_r\right)^\transp \left. \nabla_{\mathbf{x}} \nabla_{\mathbf{x}}^\transp \psi\left(\mathbf{x}\right) \right|_{\mathbf{x}_r} \left(\mathbf{x}-\mathbf{x}_r\right)$. Letting $\psi\left(\mathbf{x}\right) = \norm{\mathbf{x}-\mathbf{a}_{l,j}}{2}$ and recalling that $\forall \mathbf{x} \in \mathbb{R}^3, \nabla_{\mathbf{x}} \nabla_{\mathbf{x}}^\transp \norm{\mathbf{x}}{2} = \mathbf{Id}_2/\norm{\mathbf{x}}{2}-\mathbf{x}\mathbf{x}^\transp/\norm{\mathbf{x}}{2}^3$ yields:
    \begin{align}
        e_1 &=  \dfrac{1}{2}\left(\dfrac{\norm{\mathbf{x}-\mathbf{x}_r}{2}^2}{\norm{\mathbf{x}_r-\mathbf{a}_{l,j}}{2}}-\dfrac{\left(\mathbf{x}-\mathbf{x}_r\right)^\transp \mathbf{\Gamma}_{\mathbf{x}_r}\left(\mathbf{x}-\mathbf{x}_r\right)}{\norm{\mathbf{x}_r-\mathbf{a}_{l,j}}{2}^3}\right)\nonumber\\
        &= \dfrac{1}{2}\left(\dfrac{\norm{\mathbf{x}-\mathbf{x}_r}{2}^2}{\norm{\mathbf{x}_r-\mathbf{a}_{l,j}}{2}}-o\left(\dfrac{1}{\norm{\mathbf{x}_r - \mathbf{a}_{l,j}}{2}^2}\right)\right),
    \end{align}
    when $\left(\mathbf{x}-\mathbf{x}_r\right)^\transp \mathbf{\Gamma}_{\mathbf{x}_r}\left(\mathbf{x}-\mathbf{x}_r\right)/\norm{\mathbf{x}_r - \mathbf{a}_{l,j}}{2} \rightarrow 0$. 
    The same approach applied to $\psi\left(\mathbf{a}_{l,j}\right) = \norm{\mathbf{a}_{l,j}-\mathbf{x}_{r}}{2}$ yields:
    \begin{align}
        e_2 &= \dfrac{1}{2}\left(\dfrac{\norm{\mathbf{a}_{l,j}-\mathbf{a}_{l,r}}{2}^2}{\norm{\mathbf{x}_r-\mathbf{a}_{l,r}}{2}}-\dfrac{\left(\mathbf{a}_{l,j}-\mathbf{a}_{l,r}\right)^\transp \mathbf{\Gamma}_{\mathbf{a}_r}\left(\mathbf{a}_{l,j}-\mathbf{a}_{l,r}\right)}{\norm{\mathbf{x}_r-\mathbf{a}_{l,r}}{2}^3}\right)\nonumber\\
        &= \dfrac{1}{2}\left(\dfrac{\norm{\mathbf{a}_{l,j}-\mathbf{a}_{l,r}}{2}^2}{\norm{\mathbf{x}_r-\mathbf{a}_{l,r}}{2}}-o\left(\dfrac{1}{\norm{\mathbf{x}_r - \mathbf{a}_{l,r}}{2}^2}\right)\right),
    \end{align}
    when $\left(\mathbf{a}_{l,j}-\mathbf{a}_{l,r}\right)^\transp \mathbf{\Gamma}_{\mathbf{a}_r}\left(\mathbf{a}_{l,j}-\mathbf{a}_{l,r}\right)/\norm{\mathbf{x}_r-\mathbf{a}_{l,r}}{2} \rightarrow 0$. Considering $e=e_1+e_2$ concludes the proof.
\flushcolsend
\section{Proof of Proposition~\ref{proposition:channel_approx}}\label{appendix:c}
    Let $f_r \in \mathbb{R}$ such that $\forall f_k \in \mathbb{R}, f_k = \left(f_k-f_r\right)+f_r$. Eq.~\eqref{eq:virtual_sources_model} can be rewritten as:
    \begin{equation}\label{eq:app_intermedary}
        h_{j,k}\left(\mathbf{x}\right) = \sum_{l=1}^{L_p} \dfrac{\gamma_l}{\norm{\mathbf{x}-\mathbf{a}_{l,j}}{2}} \mathrm{e}^{-\mathrm{j}\frac{2\pi}{\lambda_r}\norm{\mathbf{x}-\mathbf{a}_{l,j}}{2}}\mathrm{e}^{-\mathrm{j}\frac{2\pi}{\lambda_{k-r}}\norm{\mathbf{x}-\mathbf{a}_{l,j}}{2}},
    \end{equation}
    where $\lambda_r$, resp. $\lambda_{k-r}$, is the wavelength associated to frequency $f_r$, resp. $f_k - f_r$. Let $\boldsymbol{\delta}_{\mathbf{x}} \triangleq \mathbf{x}-\mathbf{x}_{r}$ and $\boldsymbol{\delta}_{\mathbf{a}} \triangleq \mathbf{a}_{l,j}-\mathbf{a}_{l,r}$. Introducing Eq.~\eqref{eq:norm_approx} of Lemma~\ref{lemma:taylor_approx} in Eq.~\eqref{eq:app_intermedary} yields $\forall \left(\mathbf{x},\mathbf{a}_{l,j}\right) \in \mathcal{D}_{\mathbf{x}} \times \mathcal{D}_{\mathbf{a}}$:
    \begin{align}
        h_{j,k}\left(\mathbf{x}\right) \simeq \sum_{l=1}^{L_p} &\dfrac{\gamma_l}{\norm{\mathbf{x}-\mathbf{a}_{l,j}}{2}} \mathrm{e}^{-\mathrm{j}\frac{2\pi}{\lambda_{k-r}}\norm{\mathbf{x}_r-\mathbf{a}_{l,r}}{2}}\\
        &\cdot\mathrm{e}^{-\mathrm{j}\frac{2\pi}{\lambda_{k-r}}\mathbf{u}_{l,j}\left(\mathbf{x}_r\right)^\transp \boldsymbol{\delta}_{\mathbf{x}}} \mathrm{e}^{\mathrm{j}\frac{2\pi}{\lambda_{k-r}} \mathbf{u}_{l,r}\left(\mathbf{x}_r\right)^\transp \boldsymbol{\delta}_{\mathbf{a}}} \nonumber \\
        &\cdot\mathrm{e}^{-\mathrm{j}\frac{2\pi}{\lambda_{r}}\norm{\mathbf{x}_r-\mathbf{a}_{l,r}}{2}} \mathrm{e}^{-\mathrm{j}\frac{2\pi}{\lambda_{r}}\mathbf{u}_{l,j}\left(\mathbf{x}_r\right)^\transp \boldsymbol{\delta}_{\mathbf{x}}}\nonumber \\
        &\cdot\mathrm{e}^{\mathrm{j}\frac{2\pi}{\lambda_{r}} \mathbf{u}_{l,r}\left(\mathbf{x}_r\right)^\transp \boldsymbol{\delta}_{\mathbf{a}}}. \nonumber
    \end{align}
    Considering that the reference frequency $f_r$ is on the same order than the current frequency $f_k$ gives: $f_k - f_r \ll f_r \Rightarrow \lambda_r/\lambda_{k-r} \ll 1$. By definition of the validity domains, one has $\forall \mathbf{a}_{l,j} \in \mathcal{D}_{\mathbf{a}}, \norm{\boldsymbol{\delta}_{\mathbf{a}}}{2} \leq \epsilon_{\mathbf{a}}$ and $\forall \mathbf{x} \in \mathcal{D}_{\mathbf{x}}, \norm{\boldsymbol{\delta}_{\mathbf{x}}}{2} \leq \epsilon_{\mathbf{x}}$. Thus, the location and antenna differences $\boldsymbol{\delta}_{\mathbf{a}}$ and $\boldsymbol{\delta}_{\mathbf{x}}$ can be expressed as a small number of the reference wavelength in both directions, i.e. $\exists \left(\mathbf{m}_{\mathbf{a}}, \mathbf{m}_{\mathbf{x}} \right) \in \mathbb{R}^3 \times \mathbb{R}^3 \text{ st. } \boldsymbol{\delta}_{\mathbf{a}} \simeq \mathbf{m}_{\mathbf{a}} \lambda_r,  \boldsymbol{\delta}_{\mathbf{x}} \simeq \mathbf{m}_{\mathbf{x}} \lambda_r$ with $\norm{\mathbf{m}_{\mathbf{a}}}{2}$ and $\norm{\mathbf{m}_{\mathbf{x}}}{2}$ small. Then, $\forall \left(\mathbf{x},\mathbf{a}_{l,j}\right) \in \mathcal{D}_{\mathbf{x}} \times \mathcal{D}_{\mathbf{a}}$:
    \begin{align}
        \begin{cases}
            \boldsymbol{\delta}_{\mathbf{a}} \simeq \mathbf{m}_{\mathbf{a}} \lambda_r\\
            \boldsymbol{\delta}_{\mathbf{x}} \simeq \mathbf{m}_{\mathbf{x}} \lambda_r
        \end{cases} &\Rightarrow \begin{cases}
            \abs{\dfrac{\mathbf{u}_{l,r}\left(\mathbf{x}_r\right)^\transp \boldsymbol{\delta}_{\mathbf{a}}}{\lambda_{k-r}}} \ll 1\\
            \abs{\dfrac{\mathbf{u}_{l,j}\left(\mathbf{x}_r\right)^\transp\boldsymbol{\delta}_{\mathbf{x}}}{\lambda_{k-r}}} \ll 1
        \end{cases}.
    \end{align}
    Furthermore, as in classical communication systems the inter-antenna spacing is on the order of the wavelength, one can use the following approximation: $\forall \mathbf{x}_r \in \mathbb{R}^3, \mathbf{u}_{l,j}\left(\mathbf{x}_r\right) \simeq \mathbf{u}_{l,r}\left(\mathbf{x}_r\right)$. Namely that the DoD towards the reference location $\mathbf{x}_r$, for each antenna, can be approximated as the DoD of the reference antenna. The limitations of this approximation are discussed in this paper. Introducing $d_{l,r} \triangleq \norm{\mathbf{x}_r - \mathbf{a}_{l,r}}{2}$, one then obtain, $\forall \left(\mathbf{x},\mathbf{a}_{l,j}\right) \in \mathcal{D}_{\mathbf{x}} \times \mathcal{D}_{\mathbf{a}}$:
    \begin{align}\label{eq:app_intermedary_2}
        h_{j,k}\left(\mathbf{x}\right) \simeq \sum_{l=1}^{L_p}& \dfrac{\gamma_l}{\norm{\mathbf{x}-\mathbf{a}_{l,j}}{2}}  \mathrm{e}^{-\mathrm{j}\frac{2\pi}{\lambda_{r}}d_{l,r}} \mathrm{e}^{-\mathrm{j}\frac{2\pi}{\lambda_{r}}\mathbf{u}_{l,r}\left(\mathbf{x}_r\right)^\transp \boldsymbol{\delta}_{\mathbf{x}}}\\
        &\cdot \mathrm{e}^{-\mathrm{j}\frac{2\pi}{\lambda_{k-r}}d_{l,r}}  \mathrm{e}^{\mathrm{j}\frac{2\pi}{\lambda_{r}}\mathbf{u}_{l,r}\left(\mathbf{x}_r\right)^\transp \boldsymbol{\delta}_{\mathbf{a}}}. \nonumber
    \end{align}
    Additionally, as $\norm{\mathbf{x}-\mathbf{a}_{l,j}}{2}$ is a slowly varying term, one can approximate it as $d_{l,r}$ without suffering from a significant approximation error. One can then simplify the first two terms of Eq.~\eqref{eq:app_intermedary_2} as $h_{l,r}\left(\mathbf{x}_r\right) \triangleq \mathrm{e}^{-\mathrm{j}\frac{2\pi}{\lambda_r}d_{l,r}}/d_{l,r}$. Finally, letting $\tau_{l,r} \triangleq d_{l,r}/c$ yields $2\pi d_{l,r}/\lambda_{k-r} = 2\pi\left(f_k-f_r\right) \tau_{l,r}$. One finally obtains: $\forall \left(\mathbf{x},\mathbf{a}_{l,j}\right) \in \mathcal{D}_{\mathbf{x}} \times \mathcal{D}_{\mathbf{a}}$:
    \begin{align}
        h_{j,k}\left(\mathbf{x}\right) \simeq \sum_{l=1}^{L_p}& \gamma_l h_{l,r}\left(\mathbf{x}_r\right) \mathrm{e}^{-\mathrm{j}\frac{2\pi}{\lambda_r}\mathbf{u}_{l,r}\left(\mathbf{x}_r\right)^\transp\left(\mathbf{x}-\mathbf{x}_r\right)}\\
        &\cdot \mathrm{e}^{-\mathrm{j}2\pi\left(f_k-f_r\right) \tau_{l,r}} \mathrm{e}^{\mathrm{j}\frac{2\pi}{\lambda_r}\mathbf{u}_{l,r}\left(\mathbf{x}_r\right)^\transp \left(\mathbf{a}_{l,j}-\mathbf{a}_{l,r}\right)}. \nonumber
    \end{align}

\section{Proof of Theorem~\ref{theorem:global_validity_approx}}\label{appendix:d}
    Let us consider the tiling of the location subset $\mathcal{S}_{\mathbf{x}} \subset \mathbb{R}^3$ into $ \Omega_{\mathbf{x}}$ location validity domains $\mathcal{D}_{\mathbf{x},i}$ with the Voronoi region of any given lattice. Note that for any subset of $\mathbb{R}^3$, the optimal lattice is the $D_3$-lattice. The same approach is applied for the antenna location subset $\mathcal{S}_{\mathbf{a}} \subset \mathbb{R}^3$, yielding $ \Omega_{\mathbf{a}}$ local validity domains $\mathcal{D}_{\mathbf{a},i}$. By definition of the local validity domains and of the Taylor expansion, one has, $\forall \left(\mathcal{D}_{\mathbf{x},i},\mathcal{D}_{\mathbf{a},i}\right)$: 
    \begin{equation}
        \forall \left(\mathbf{x},\mathbf{a}_{l,j}\right) \in \mathcal{D}_{\mathbf{x},i} \times \mathcal{D}_{\mathbf{a},i}, \norm{\mathbf{H}\left(\mathbf{x}\right)-\hat{\mathbf{H}}\left(\mathbf{x}\right)}{F} < \epsilon,
    \end{equation}
    where $\hat{\mathbf{H}}\left(\mathbf{x}\right)$ is the Taylor-approximated channel matrix computed using Eq.~\eqref{eq:matrix_chan_approx}. For each local validity domain $\left(\mathcal{D}_{\mathbf{x},i},\mathcal{D}_{\mathbf{a},i}\right)$  pair, Eq.~\eqref{eq:matrix_chan_approx} shows that the channel can be approximated using only $L_p$ planar wavefronts, SVs, and FRVs. Thus, one can construct a dictionary of planar wavefronts $\tilde{\boldsymbol{\psi}}_{\mathbf{x}} \in \mathbb{C}^D$, a dictionary of SVs $\tilde{\mathbf{\Psi}}_{\mathbf{a}} \in \mathbb{C}^{N_a \times D}$, and a dictionary of FRVs $\tilde{\mathbf{\Psi}}_{\mathbf{f}} \in \mathbb{C}^{N_s \times D}$ containing the needed atoms for every local validity domain pair. This yields $D \leq L_p \Omega_{\mathbf{x}} \Omega_{\mathbf{a}}$. 
    
    While the FRV dictionary atoms follow Eq.~\eqref{eq:freq_correction}, the SV dictionary atoms can't be constituted from Eq.~\eqref{eq:ant_correction}, as only the true antenna locations $\mathbf{a}_{1,j}$ are known. However, every virtual antenna can be rewritten as a rotated and translated version of its physical counterpart, i.e. $\forall l\neq 1,\exists! \left(\mathbf{R}_{\theta,l},\boldsymbol{\epsilon}_l\right) \in \mathbb{R}^{3 \times 3} \times  \mathbb{R}^{3}, \text{ s.t. } \mathbf{a}_{l,j} = \mathbf{R}_{\theta,l} \mathbf{a}_{1,j} + \boldsymbol{\epsilon}_l$. One then obtains:
    \begin{align}
        \left(\mathbf{a}_{l,j}-\mathbf{a}_{l,r}\right) &= \mathbf{R}_{\theta,l}\mathbf{a}_{1,j}+\boldsymbol{\epsilon}_l - \mathbf{R}_{\theta,l}\mathbf{a}_{1,r}-\boldsymbol{\epsilon}_l\nonumber\\
        &= \mathbf{R}_{\theta,l} \left(\mathbf{a}_{1,j}-\mathbf{a}_{1,r}\right).
    \end{align}
    While the antenna difference is not impacted by translations, it is impacted by rotations. However, one can easily show that the projection term in the SV is rotation equivariant:
    \begin{align}
        \mathbf{u}_{l,r}\left(\mathbf{x}_r\right)^\transp \left(\mathbf{a}_{l,j}-\mathbf{a}_{l,r}\right) &= \mathbf{u}_{l,r}\left(\mathbf{x}_r\right)^\transp \mathbf{R}_{\theta,l} \left(\mathbf{a}_{1,j}-\mathbf{a}_{1,r}\right)\nonumber\\
        &= \left(\mathbf{R}_{\theta,l}^\transp\mathbf{u}_{l,r}\left(\mathbf{x}_r\right)\right)^\transp\left(\mathbf{a}_{1,j}-\mathbf{a}_{1,r}\right)\nonumber\\
        &= \tilde{\mathbf{u}}_{l,r}\left(\mathbf{x}_r\right)^\transp\left(\mathbf{a}_{1,j}-\mathbf{a}_{1,r}\right).
    \end{align}
    Thus, the dictionary of SVs can be constructed using DoDs $\tilde{\mathbf{u}}_i \in \mathbb{R}^3$ and the true antenna location difference $\left(\mathbf{a}_{1,j}-\mathbf{a}_{1,r}\right)$.
    Using the previously defined dictionaries and introducing an activation vector $\mathbf{w}\left(\mathbf{x}\right) \in \mathbb{C}^D$ to select the needed $L_p$ atoms at the current location $\mathbf{x}$, such that $\boldsymbol{\varpi}\left(\mathbf{x}\right) = \mathbf{w}\left(\mathbf{x}\right) \odot  \tilde{\boldsymbol{\psi}}_{\mathbf{x}}\left(\mathbf{x}\right)$, yields Eq.~\eqref{eq:global_approx} and concludes the proof.
    
\end{document}